\documentclass[10pt]{amsart}

\usepackage{amssymb}

\newtheorem{thm}{Theorem}
\newtheorem{lem}[thm]{Lemma}
\newtheorem{cor}[thm]{Corollary}
\newtheorem{prop}[thm]{Proposition}

\def\sign{{\rm sign}}
\def\less{\lesssim}
\def\la{\langle}
\def\ra{\rangle}

\def\Ai{{\rm Ai}}
\def\Bi{{\rm Bi}}
\def\eps{{\varepsilon}}
\def\les{\lesssim}

\def\beeq{\begin{equation}}
\def\eneq{\end{equation}}
\def\R{{\mathbb R}}
\def\tilde{\widetilde}
\def\wt{\widetilde}
\renewcommand{\Sigma}{{\mathbb S}}

\begin{document}

\numberwithin{equation}{section}
\title[Zero energy scattering]
{Semiclassical analysis of low and zero energy scattering for one dimensional Schr\"odinger operators with inverse square potentials}
\author{Ovidiu Costin, Wilhelm Schlag, Wolfgang Staubach, Saleh Tanveer}
\thanks{The first author was supported by the National Science Foundation DMS--0406193, DMS--0601226, DMS--0600369, the second by  DMS--0653841, and the 
fourth by  DMS--0405837. The second author thanks Fritz Gesztesy and Gerald Teschl for their interest in this work and comments on a preliminary version.}
\address{Costin, Tanveer: Department of Mathematics, The Ohio State University 100 Math Tower, 231 West 18th Avenue, Columbus, OH 43210-1174, U.S.A.} 
\email{costin@math.ohio-state.edu,  tanveer@math.ohio-state.edu}
\address{Schlag: The University of Chicago, 5734 South University Avenue, Chicago, IL 60637, U.S.A.}
\email{schlag@math.uchicago.edu} 
\address{Staubach: Department of Mathematics, Colin Maclaurin Building, Heriot-Watt University, Edinburgh,  EH14 4AS, U.K.}
\email{W.Staubach@hw.ac.uk}
\maketitle

\begin{abstract}
This paper studies the scattering matrix $\Sigma(E;\hbar)$ of the problem
\[
 -\hbar^2 \psi''(x) + V(x) \psi(x) = E\psi(x)
\]
for positive potentials $V\in C^\infty(\R)$ with inverse square behavior as $x\to\pm\infty$.
It is shown that each entry takes the form $\Sigma_{ij}(E;\hbar)=\Sigma_{ij}^{(0)}(E;\hbar)(1+\hbar \sigma_{ij}(E;\hbar))$ where $\Sigma_{ij}^{(0)}(E;\hbar)$
is the WKB approximation relative to the {\em modified potential} $V(x)+\frac{\hbar^2}{4} \la x\ra^{-2}$  and the correction terms $\sigma_{ij}$
satisfy $|\partial_E^k \sigma_{ij}(E;\hbar)| \le C_k E^{-k}$ for all $k\ge0$ and uniformly in $(E,\hbar)\in (0,E_0)\times (0,\hbar_0)$
where $E_0,\hbar_0$ are small constants. This asymptotic behavior is not universal: if $-\hbar^2\partial_x^2 + V$ has a {\em zero energy resonance}, then  $\Sigma(E;\hbar)$ exhibits different asymptotic behavior as $E\to0$.
The resonant case is excluded here due to $V>0$.
\end{abstract}

\section{Introduction}

This paper revisits the much studied problem of determining
the reflection and transmission coefficients for semi-classical operators  of the form
\begin{equation}
 \label{eq:semiclass} P(x,\hbar D) := -\hbar^2 \frac{d^2}{dx^2} + V(x)
\end{equation}
where $V$ is real-valued and assumed to decay at infinity.  There are two atypical features of this work, at least relative
to the existing literature on this topic:
\begin{enumerate}
 \item[(i)]  we wish to understand
the zero energy limit, in fact uniformly\footnote{More precisely, the asymptotic analysis is carried out up to multiplicative errors of the form
$1+O(\hbar)$ where the $O(\hbar)$ needs to be uniform in small energies.
} in small $\hbar$
\item[(ii)] the smooth potential $V$ decays like an inverse square at both ends\footnote{The methods of this paper also
apply to the case where the potential exhibits  inverse square decay as $x\to\infty$ and some other decay as $x\to-\infty$; for that,
one of course needs to be able to carry out the scattering theory on $x<0$. If the decay is $|x|^{-\alpha}$ with $0<\alpha<2$, then \cite{Yaf} applies, whereas for $\alpha>2$ one can use classical scattering methods.}
\end{enumerate}
We remark that (i) and (ii) are closely related. Indeed, the $\la x\ra^{-2}$ decay is ``critical'' with respect to the
zero energy limit in the sense that $\la x\ra^{-2-\eps}$ is easier and behaves very differently.
In the semi-classical literature it is more customary to encounter the criticality of the Coulomb decay $\la x\ra^{-1}$;
the reason for this is that the Coulomb decay is critical for  {\em positive energies}. Note that the numerology around
these decay rates applies to all dimensions and not just to one dimension. The motivation for considering this particular
problem comes from several sources. First, smooth potentials which behave as an inverse square at one or both ends arise in several
contexts in physics and geometry, for example in general
relativity in connection with Schwarzschild and de-Sitter spaces, see~\cite{Chan}. 
Second, this paper is part of 
the program initiated in~\cite{SSS1} and~\cite{SSS2}. In fact, the analysis carried out here is an essential part in the
solution of the  ``large angular momentum'' problem from~\cite{SSS2}.

Let us briefly review some elementary features of scattering, cf.~\cite{DT} and~\cite{M}:
For simplicity, let $\hbar=1$ for now and write $H=P(x,D)$. Recall that the Jost solutions $f_{\pm}(x;\lambda)$ are required
to satisfy
\[ H f_{\pm}(\cdot,\lambda)=\lambda^2 f_{\pm}(\cdot,\lambda),\qquad
f_{\pm}(x,\lambda)\sim e^{\pm i\lambda x} \text{\ \ as\ \ }x\to\pm\infty
\]
Provided $V\in L^1$ and $\lambda\ne0$ they exist and are uniquely determined as solutions of the Volterra equation
\begin{equation}\label{eq:Volterra}
 f_{+}(x,\lambda) = e^{ix\lambda} + \int_x^\infty \frac{\sin(\lambda(y-x))}{\lambda} V(y) f_+(y,\lambda)\,dy
\end{equation}
and similarly for $f_-(\cdot,\lambda)$. The resolvent kernel of $H$ can now be expressed in the form
\[
 (H-(\lambda^2+i0))^{-1}(x,y) = \frac{f_+(x,\lambda) f_-(y,\lambda)}{W(\lambda)} \chi_{[x>y]} + \frac{f_+(y,\lambda) f_-(x,\lambda)}{W(\lambda)} \chi_{[x<y]}
\]
for all $\lambda>0$
where $W(\lambda)=W(f_+(\cdot,\lambda) ,f_-(\cdot,\lambda))$. The reflection and transmission coefficients are
defined by the relations
\begin{align*}
  t_+(\lambda) f_+(\cdot,\lambda) &= r_+(\lambda) f_-(\cdot,\lambda) + f_-(\cdot,-\lambda) \\
  t_-(\lambda) f_-(\cdot,\lambda) &= r_-(\lambda) f_+(\cdot,\lambda) + f_+(\cdot,-\lambda)
\end{align*}
and satisfy
\begin{align}\label{eq:rt_rel}
 t_- &= t_+,\quad 1 = |t_+|^2 + |r_+|^2 = |t_-|^2 + |r_-|^2, \qquad r_- = -\bar r_+{t}/\,{\bar t}
\end{align}
For fixed $\lambda>0$, consider the following bases of the space of solutions to the equation $Hf=\lambda^2 f$:
 \[ (f_+(\cdot,\lambda),f_-(\cdot,\lambda)),\qquad (f_+(\cdot,-\lambda),f_-(\cdot,-\lambda))\]
The former is referred to  as {\em outgoing} and the latter
 as {\em incoming}. In that case
the matrix $\Sigma(\lambda)$ which transforms the coefficients of a solution relative to these bases satisfies
\[
 \Sigma(\lambda) = \left [ \begin{matrix}  t(\lambda) &  r_-(\lambda) \\ r_+(\lambda) & t(\lambda) \end{matrix} \right ]
\]
It is called the {\em scattering matrix} and is unitary.
Of special interest to us is the behavior as $\lambda\to0+$.
Note that if \[ \int_{-\infty}^\infty \la x\ra |V(x)|\,dx<\infty \]
then $f_+(x,\lambda)\to f_+(x,0)$ as $\lambda\to0$ where the latter satisfies the limit equation of~\eqref{eq:Volterra}, viz.
\[
 f_{+}(x,0) = 1 + \int_x^\infty (y-x) V(y) f_+(y,0)\,dy
\]
It is known  that $\Sigma(\lambda)$ is continuous in $\lambda\ge0$ under this moment condition, see~\cite{K}.  To describe the possible values of $\Sigma(0)$, recall that
$H$ has a {\em zero energy resonance} iff $f_{\pm}(\cdot,0)$ are linearly dependent or, equivalently,
iff $W(0)=0$. Furthermore, since $t(\lambda)=-\frac{2i\lambda}{W(\lambda)}$ this is equivalent to $t(0)\ne0$.
In conclusion, if zero energy is not resonant, then
\[
 \Sigma(0)= \left [ \begin{matrix}  0 & -1 \\ -1 & 0 \end{matrix} \right ]
\]
whereas in the resonant case
\[
 \Sigma(0) =  \left [ \begin{matrix}  t & -r \\ r & t \end{matrix} \right ]
\]
for some real $r,t\in [-1,1]$, $t\ne0$.

\noindent  
If $\la x\ra V(x)\not\in L^1(\R)$, then the behavior of $\Sigma(\lambda)$ as $\lambda\to0$ is completely different. 
In this paper, we focus on the border line case of positive inverse square potentials for~\eqref{eq:semiclass} and $\hbar$ small
(for the remainder of he paper, we now let $\hbar$ be a small positive quantity). 
It is precisely this case which arises in the geometric problem considered in~\cite{SSS1}, \cite{SSS2}.
 Our main theorem is as follows. 
We denote the energy by $E=\lambda^2>0$, see above, and the scattering matrix of~\eqref{eq:semiclass} by
\[
 \Sigma(E;\hbar) = \left [ \begin{matrix}  t(E;\hbar) &  r_-(E;\hbar) \\ r_+(E;\hbar) & t(E;\hbar) \end{matrix} \right ]
= \left [ \begin{matrix} \Sigma_{11}(E;\hbar) & \Sigma_{12}(E;\hbar) \\ \Sigma_{21}(E;\hbar) & \Sigma_{22}(E;\hbar) \end{matrix} \right ]
\]
In view of \eqref{eq:rt_rel} it suffices to describe the first row of this matrix.  In this paper, $O(\cdot)$ terms  will
be differentiable functions and we will typically state bounds on their derivatives with regard to the relevant variables
depending on the context. 

\begin{thm}\label{thm:main}
Let $V\in C^\infty(\R)$ with $V>0$ and $V(x)=\mu_{\pm}^2 x^{-2} + O(x^{-3})$
as $x\to\pm\infty$ where $\mu_+\ne0$, $\mu_-\ne0$
and $\partial_x^k O(x^{-3}) = O(x^{-3-k})$ for all $k\ge0$.
Denote
\begin{equation}\label{eq:V0_def}
 V_0(x;\hbar):= V(x) + \frac{\hbar^2}{4}\la x\ra^{-2}
\end{equation}
and let  $E_0>0$ be such
that for all $0<E<E_0$ and $0<\hbar<1$, $V_0(x;\hbar)=E$ has a unique pair of solutions,
which we denote by $x_2(E;\hbar)<0<x_1(E;\hbar)$. Define
\begin{equation}\label{eq:STdef} \begin{split}
S(E;\hbar) &:= \int_{x_2(E;\hbar)}^{x_1(E;\hbar)} \sqrt{V_0(y;\hbar)-E}\, dy \\
 T_+(E;\hbar) &:= x_1(E;\hbar)\sqrt{E} - \int_{x_1(E;\hbar)}^\infty \big(\sqrt{E-V_0(y;\hbar)}-\sqrt{E}\big)\, dy \\
 T_-(E;\hbar) &:= -x_2(E;\hbar)\sqrt{E} -  \int_{-\infty}^{x_2(E;\hbar)} \big(\sqrt{E-V_0(y;\hbar)}-\sqrt{E}\big)\, dy
\end{split}\end{equation}
as well as $T(E;\hbar):= T_+(E;\hbar) + T_-(E;\hbar)$.
Then for all $0<\hbar<\hbar_0$ where  $\hbar_0=\hbar_0(V)>0$ is small and $0<E<E_0$
\begin{equation}\label{eq:sentries}
\begin{aligned}
 \Sigma_{11}(E;\hbar) &= e^{-\frac{1}{\hbar}(S(E;\hbar)+iT(E;\hbar))} (1+ \hbar\,\sigma_{11}(E;\hbar)) \\
 \Sigma_{12} (E;\hbar)  &= -i e^{-\frac{2i}{\hbar} T_+(E;\hbar)} (1+\hbar\, \sigma_{12}(E;\hbar))
\end{aligned}
\end{equation}
where the correction terms satisfy the bounds
\begin{equation}
\label{eq:errors}
 |\partial_E^k\, \sigma_{11}(E;\hbar)|+|\partial_E^k\, \sigma_{12}(E;\hbar)| \le C_k\, E^{-k}\quad\forall\; k\ge0,
\end{equation}
with a constant $C_k$ that only depends on $k$ and $V$. The same conclusion holds if instead of~\eqref{eq:V0_def} we
were to define $V_0$ as $V_0:=V+\hbar^2 V_1$ with $V_1\in C^\infty(\R)$, 
$V_1(x;\hbar)=\frac14\la x\ra^{-2}+O(x^{-3})$ as $x\to\pm\infty$ with $\partial_x^k O(x^{-3}) = O(x^{-3-k})$ for all $k\ge0$
and uniformly in~$0<\hbar\ll1$. 
\end{thm}

The addition of $\frac{\hbar^2}{4}\la x\ra^{-2}$ to $V(x)$ is crucial and similar to the ``Langer modification'', see for example~\cite{FT}; indeed, if we were to use $V$ instead of $V_0$ in~\eqref{eq:STdef},
then the bounds \eqref{eq:errors} would fail due to a factor of $\log E$ as $E\to0$. 
This is in contrast to potentials decaying like $|x|^{-\alpha}$ with $0<\alpha<2$ for which the modification is not needed, i.e., the
usual WKB ansatz works, see~\cite{Yaf}. 
On the other
hand, note that as long as $E_0>E>\eps>0$ the turning points $x_j(E;\hbar)$ will remain bounded and the distinction between
$V_0$ and $V$ is therefore moot. Indeed, the effect of passing from $V$ to $V_0$ and vice versa is merely a harmless factor of
the form $1+O(\hbar)$ where the $O(\cdot)$ term of course depends on~$\eps$. In the range $E_0>E>\eps>0$ Theorem~\ref{thm:main}
is well-known and classical. See for example Chapter~13 of~\cite{Olver} as well as Ramond's work~\cite{Ramond} for a more recent reference (Ramond, however, is more concerned with the scattering problem for energies close to the maximum of a barrier and he also assumes that
the potential is dilation analytic).

We remark that the infinite differentiability assumption on $V$ can be relaxed to some finite amount of smoothness (in which
case we can  only ask for correspondingly many derivatives with respect to~$E$), but we do not elaborate on this issue here.
A more substantial problem is that of relaxing the positivity assumption.
We conjecture that  $V>0$ can be replaced by the strictly weaker assumption that {\em zero energy is not a resonance} of $P(x,\hbar D)$. Recall
the definition of a zero energy resonance in this context, cf.~\cite{BGW}, \cite{Yaf}, and Section~3 of~\cite{SSS2}: it means that the two subordinate zero-energy solutions at $\pm\infty$ are linearly dependent (a ``subordinate solution'' at either end refers
to the nonzero solution of $P(x,\hbar D)f=0$ with the slowest possible growth at that end; it is unique up to a nonzero scalar factor).

Note, however, that some condition is needed in Theorem~\ref{thm:main}; indeed, in \cite{SSS2} 
it was shown that for operators of the form considered
in Theorem~\ref{thm:main}  with $\mu_+^2=\mu_-^2=\nu^2-\frac14$, $\nu>\frac12$, and $\hbar=1$ 
\begin{equation}\label{eq:Wlam_asymp}
 W(E;\hbar) \sim  E^{\frac12-\nu} ( W_0 + O(E^\eps)) \text{\ \ as\ \ }E\to0+
\end{equation}
 for some $W_0\ne0$ and $\eps>0$ 
{\em provided there is no zero energy resonance}. In the resonant case, it was shown in~\cite{SSS2} that $W_0=0$. 
The following relation  between $\Sigma_{11}$ in~\eqref{eq:sentries} and the Wronskian $W(E;\hbar)$
\begin{equation}\label{eq:Sigma11W} W(E;\hbar)=\frac{-2i\sqrt{E}}{\hbar \Sigma_{11}(E;\hbar)}\end{equation} 
allows one to deduce~\eqref{eq:Wlam_asymp} with $W_0\ne0$ from Theorem~\ref{thm:main} (note that for inverse square potentials $S(E;\hbar)$
behaves like $|\log E|$ so that the apparent exponential behavior in~\eqref{eq:sentries} turns into a power-law in~$E$). 
This deduction also proves that Theorem~\ref{thm:main} necessarily fails in the presence of a zero energy resonance. 
Another aspect of~\eqref{eq:Wlam_asymp} concerns the case of {\em large}~$\hbar$, say $\hbar=1$. Indeed, it shows that Theorem~\ref{thm:main} gives
the correct behavior of the scattering matrix even in that case, but then the energy takes over as the small parameter. 

This paper is organized as follows. Section~\ref{sec:zero} constructs a fundamental system of zero energy solutions
to~\eqref{eq:semiclass} via the usual WKB ansatz but for $V_0$ rather than for~$V$. Since we require uniform bounds in 
Theorem~\ref{thm:main} as $E\to0$, the construction of Jost solutions for positive energies which is carried out in 
Section~\ref{sec:Langer} needs to yield the zero energy solutions in the limit $E\to0$. We choose to
reverse this process and show that $V_0$ is precisely the right potential to use in the WKB method at zero energy. 
The logic is simple: the WKB ansatz \[\psi_{0,\pm}(x):=V^{-\frac14}(x) \exp\big(\pm\hbar^{-1}\int_{x_0}^x \sqrt{V(y)}\,dy\big)\]
satisfies an equation of the form 
\[
 (-\hbar^2\partial_x^2 + V)\psi_{0,\pm} = \hbar^2 (-x^{-2}/4 + O(x^{-3})) \psi_{0,\pm}
\]
where the $x^{-2}/4 $ term on the right-hand side is universal for all potentials that have an inverse square
decay as $x\to\infty$ as specified in Theorem~\ref{thm:main}. 
Since this term  has the same decay as~$V$ we need to bring it to the left-hand side leading to our choice of~$V_0$. 

The main technical work of this paper is carried out in Section~\ref{sec:Langer}. It is here that the (semi-classical)
Jost solutions are constructed for all energies in the range $0<E<E_0$. We use Langer's method which is based on the
Liouville-Green transform, see Chapters~6 and~11 in~\cite{Olver}: switching to the new independent variable
\[
 \zeta=\zeta(x,E;\hbar):= \sign (x-x_1(E;\hbar)) \Big|\frac32 \int_{x_1(E;\hbar)}^x
\sqrt{|V_0(x;\hbar)-E|}\, d\eta\Big|^{\frac23},\qquad x\ge0
\]
and to the new dependent variable $w(\zeta)=\sqrt{\zeta'}f$ reduces $P(x,\hbar D)f=Ef$, see~\eqref{eq:semiclass}, to an Airy equation perturbed by a potential of size $\hbar^2$.  It is here that $V>0$ becomes relevant: it ensures that for all small $E>0$ there is a unique
turning point $x_1(E)>0$ and that $V_0(x;\hbar)>E$ for all $0<x<x_1(E;\hbar)$. Hence we can cover $x\ge0$ by the intervals $\zeta(0,E;\hbar)<\zeta\le 0$ and $\zeta\ge0$. In each of these intervals we solve the perturbed Airy equations up to 
multiplicative errors of the form $1+O(\hbar)$ where the $O(\cdot)$ term is uniform in~$E$. It is in the range $\zeta(0,E;\hbar)<\zeta\le 0$
that the choice of $V_0$ (rather than $V$) becomes decisive; this of course is to be expected as this range turns into 
the whole interval $x\ge0$ as $E\to0$ and WKB applied to~$V$ instead of~$V_0$ fails at $E=0$, see Section~\ref{sec:zero}. 
Theorem~\ref{thm:main} is proved in Section~\ref{sec:proof} by evaluating the Wronskians
\[
 W(f_+(\cdot,E), f_-(\cdot,E)),\qquad W(f_+(\cdot,E), \overline{f_-(\cdot,E)})
\]
at $x=0$.  Section~\ref{sec:largeE} discusses the range of validity of Theorem~\ref{thm:main} as the energy
increases towards a unique non-degenerate maximum of a barrier potential. Finally, the appendix describes a certain ``normal-form'' reduction of \eqref{eq:semiclass} to a Bessel equation
on a region containing the turning point. Even though we do not base our asymptotic analysis on this reduction (but rather
the Airy equation), we
still believe that this is of independent interest. 

Needless to say, there is a vast literature related to the semi-classical analysis of the Schr\"odinger equation
and it is impossible to do any justice to it here. Somewhat curiously, however, there does not seem to be
any  literature on potentials which are globally smooth on the line and which exhibit inverse square decay. 
On the other hand, potentials which are {\em exactly} inverse square are of course ubiquitous, especially in the physics
literature. For a recent paper in this direction involving WKB see~\cite{FT} and for a time-dependent analysis
see the recent papers~\cite{PST1}, \cite{PST2},
as well as~\cite{BPST1}, \cite{BPST2} and the references cited there. Potentials which decay of the form $|x|^{-\alpha}$, $0<\alpha<2$, 
have been studied with similar objectives as here, see~\cite{Kvi}, \cite{Nak}
and~\cite{Yaf}. For other work on low energies see \cite{BGS}, \cite{BGW}, \cite{DS}, and~\cite{Yaf2}, as well as~\cite{GH}.

\section{Zero energy solutions}\label{sec:zero}

In order to motivate the choice of $V_0$ in Theorem~\ref{thm:main} we will now obtain
a fundamental system for the equation

\begin{equation}                                                                    \label{eq1}
-\hbar^{2} f''(x) + V(x)f(x) = 0
\end{equation}
on the half axis $x>x_0$. Here we assume that $V(x)=\mu^2 x^{-2} +O(x^{-3})$ with $\mu>0$ as $x\to\infty$ and $x_0$ is chosen so
large that $V(x)>0$ for $x>x_0$.  As before, we require $\partial_x^k O(x^{-3}) = O(x^{-3-k})$ for all $k\ge0$.
Ignoring the $O(\cdot)$ term, we have on the one hand
\[
 \big(-\hbar^{2}\partial_x^2 + \mu^2 x^{-2}\big) x^{\frac12 \pm \alpha}= 0,\quad \alpha^2=\frac14+\mu^2\hbar^{-2}
\]
On the other hand, with $Q(x;\hbar):=\mu^2 x^{-2} +\hbar^2 x^{-2}/4$,
\[
 Q^{-\frac14}(x;\hbar) e^{\pm\frac{1}{\hbar}\int_{x_0}^x \sqrt{Q(y;\hbar)}\,dy} = c x^{\frac12 \pm\alpha}
\]
with some $c\ne0$.
This motivates the following result.

\begin{prop}\label{zeroenergi thm}
On $x>x_0$ a fundamental system of solutions for \eqref{eq1} is given by
\begin{equation}                                                                        \label{eq2}
\begin{split}
\psi_{j}(x;\hbar)=\,& \tilde{\psi}_{j}(x;\hbar) (1+h a_{j}(x;\hbar)),\quad j=1,2
\end{split}
\end{equation}
with
\begin{equation}                                                                  \label{jost}
\begin{split}
\tilde{\psi}_{1}(x;\hbar)=& \,V_0(x;\hbar)^{\frac{-1}{4}}e^{\frac{1}{\hbar}S(x;\hbar)}\\
\tilde{\psi}_{2}(x;\hbar)=& \,V_0(x;\hbar)^{\frac{-1}{4}}e^{\frac{-1}{\hbar}S(x;\hbar)},
\end{split}
\end{equation}
where $V_0(x;\hbar)=V(x)+ \frac{\hbar^{2}}{4}\langle x\rangle ^{-2}$, $S(x;\hbar)=\int_{x_{0}}^{x} \sqrt{ V_0(t;\hbar)}\,dt$
and
\begin{equation}                                                                                \label{eq3}
\sup_{0<\hbar<1}\vert \partial^{\ell}_{x}a_{j}(x;\hbar)\vert \le C_{\ell,\,\mu}\, x^{-\ell}
\end{equation}
for $x>x_{0}$, $j=1, 2$ and $\ell=0,1$. Their Wronskian satisfies
\begin{equation}
W(\psi_1, \psi_2)= \frac{-2}{\hbar}(1+O(\hbar)).
\end{equation}
as $\hbar\to0$.
\end{prop}

\begin{proof}
Let us consider the case of $\psi_{1}$. Hence, we need to find $a_1$ so that $\psi_{1}$ is a
solution to the differential equation
\begin{equation}                                                                         \label{eq4}
-\hbar^{2}u''(x)+V(x) u(x)=0.
\end{equation}
Substituting the first expression of \eqref{eq2} into the
differential equation \eqref{eq4} yields
\begin{equation}                                                                         \label{eq7}
\begin{split}
-\hbar^{2}[{\tilde{\psi}}''_{1}(1+\hbar a_{1}) +2\hbar{\tilde{\psi}}'_{1}a'_{1}
+\hbar {\tilde{\psi}}_{1}
a''_{1}  ] +V\tilde{\psi}_{1}(1+\hbar a_{1})=0
\end{split}
\end{equation}
Setting $V_2:= \frac{1}{4}\langle x\rangle ^{-2}-\frac{1}{4}\frac{
V_0''}{V_0}+\frac{5}{16}\frac{V_0'^2}{V_0^2}$ and observing
that
$-\hbar^{2}{\tilde{\psi}}''_{1}+V\tilde{\psi}_{1}=-\hbar^{2}V_{2}\tilde{\psi}_{1}$,
we deduce after dividing the equation by $\tilde{\psi}_{1}$
\begin{equation}                                                                          \label{eq7'}
-(1+ \hbar a_{1})V_2 = \hbar (a''_{1} +
2\frac{\tilde{\psi}'_{1}}{\tilde{\psi}_{1}}a'_{1}),
\end{equation}
We now note the following essential feature of $V_2$ (which was the reason for defining $V_0$ as above):
\[
|V_2(x)|\le C\, x^{-3},\qquad |\partial_x^k V_2(x)|\le C_k\, x^{-3-k}\qquad\forall\;k\ge0
\]
To solve \eqref{eq7'} we multiply  both sides by
$\tilde{\psi}^{2}_{1}$ and obtain
\begin{equation}\label{eq:ODE}
(a'_{1}\tilde{\psi}^{2}_{1})'=
\frac{-1}{\hbar}V_{2}\tilde{\psi}^{2}_{1}-a_{1}V_{2}\tilde{\psi}^{2}_{1}.
\end{equation}
Integration and using the definition of the
$\tilde{\psi}_{1}$ yield
\begin{equation}                                                                    \label{a1derivative}
\begin{split}
a'_{1}(x) &=
\frac{1}{\hbar}\int_{x}^{x_{0}}V_{2}(y)\tilde{\psi}^{-2}_{1}(x)\tilde{\psi}^{2}_{1}(y)\,
dy +
\int_{x}^{x_{0}}a_{1}(y)V_{2}(y)\tilde{\psi}^{-2}_{1}(x)\tilde{\psi}^{2}_{1}(y)\,
dy\\
& =
 \frac{1}{\hbar}\int_{x}^{x_{0}} V_0(x)^{\frac{1}{2}} V_0(y)^{\frac{-1}{2}}
 e^{\frac{2}{\hbar}(S(y)-S(x))}V_{2}(y) dy \\&\quad +\int_{x}^{x_{0}}
 V_0(x)^{\frac{1}{2}} V_0(y)^{\frac{-1}{2}}e^{\frac{2}{\hbar}(S(y)-S(x))}
 V_{2}(y)a_{1}(y)\, dy.
\end{split}
\end{equation}
Strictly speaking, $a_1=a_1(x,\hbar)$ but we suppress the $\hbar$ from the notation here.
After integration in \eqref{a1derivative} we obtain
\begin{equation}
\begin{split}
a_{1}(x)=
 &\frac{-1}{\hbar}\int_{x}^{x_{0}}\int_{x'}^{x_{0}} V_0(x')^{\frac{1}{2}}
 V_0(y)^{\frac{-1}{2}}e^{\frac{2}{\hbar}(S(y)-S(x'))}V_{2}(y)\, dy\, dx' \\- &
\int_{x}^{x_{0}}\int_{x'}^{x_{0}}V_0(x')^{\frac{1}{2}}
V_0(y)^{\frac{-1}{2}}e^{\frac{2}{\hbar}(S(y)-S(x'))}V_{2}(y)a_{1}(y)\, dy\, dx' \\
= & \frac{-1}{\hbar}\int_{x}^{x_{0}}\int_{x}^{y}
V_0(x')^{\frac{1}{2}}
V_0(y)^{\frac{-1}{2}}e^{\frac{2}{\hbar}(S(y)-S(x'))}V_{2}(y)\, dx'\, dy \\- &
\int_{x}^{x_{0}}\int_{x}^{y}V_0(x')^{\frac{1}{2}}
V_0(y)^{\frac{-1}{2}}e^{\frac{2}{\hbar}(S(y)-S(x'))}V_{2}(y)a_{1}(y)\,
dx'\, dy.
\end{split}
\end{equation}
Furthermore,
\begin{equation}
\int_{x}^{y} V_0(x')^{\frac{1}{2}}
e^{\frac{-2}{\hbar}S(x')}
dx'=-\frac{\hbar}{2}[e^{\frac{-2}{\hbar}S(y)}-e^{\frac{-2}{h}S(x)}].
\end{equation}
and
\begin{equation}
V_0(y)^{\frac{-1}{2}}
e^{\frac{2}{\hbar}S(y)}\int_{x}^{y}
V_0(x')^{\frac{1}{2}} e^{\frac{-2}{\hbar}S(x')} dx'
= \frac{-\hbar}{2}
V_0(y)^{\frac{-1}{2}}[1-e^{\frac{2}{\hbar}(S(y)-S(x))}].
\end{equation}
From this it follows that
\begin{equation}                                                                  \label{a1}
a_{1} (x) =
\frac{1}{2}\int_{x_{0}}^{x}V_0(y)^{\frac{-1}{2}}[e^{\frac{2}{\hbar}(S(y)-S(x))}-1]V_{2}(y)
[1+\hbar a_{1}(y)]\,dy.
\end{equation}
This is a standard Volterra equation.
To solve it, we first introduce a new
function $\rho(x)$ given by
\begin{equation}
\rho(x)= \int_{x_{0}}^{x}
\vert V_0(y)^{\frac{-1}{2}}V_{2}(y)\vert dy.
\end{equation}
In view of the decay of $V_2$ we see that the integrand here decays like $y^{-2}$ so that $\rho\in L^\infty(x_0,\infty)$.
Then we define a sequence $a^{s}_{1}(x)$, $s=0,1,\,\dots ,\,$ with
$a^{0}_{1}(x)=0$ and
\begin{equation}                                                                                \label{eq12}
a^{s}_1 (x) = \frac{1}{2}\int_{x_{0}}^{x}
V_0(y)^{\frac{-1}{2}}[e^{\frac{2}{\hbar}(S(y)-S(x))}-1]V_{2}(y)[1+\hbar a^{s-1}_{1}(y)]\,
dy
\end{equation}
We claim that
\begin{equation}                                                                                   \label{eq13}
\vert a^{s}_1 (x)-a^{s-1}_1 (x)\vert \leq
\frac{\rho^{s}(x) h^{s-1}}{s!}.
\end{equation}
To prove this we proceed by induction and observe that
\[S(y)-S(x)=-\int_{y}^{x} V_0^{\frac{1}{2}}(t,\hbar)\, dt\] and hence for $0<x_{0} <y<x,$
$\vert e^{\frac{2}{\hbar}(S(y)-S(x))}-1 \vert\leq 2$. Therefore, \eqref{eq13}
is valid for $s=1$. Furthermore, if we assume the validity of
\eqref{eq13} for $s=k$ then since
\begin{equation}                                                                                  \label{eq51}
a^{k+1}_1 (x)-a^{k}_1 (x)= \frac{\hbar}{2}\int_{x}^{x_{0}}\\
V_0(y)^{\frac{-1}{2}}[e^{\frac{2}{\hbar}(S(y)-S(x))}-1]V_{2}(y)
(a^{k}_{1}(y)-a^{k-1}_{1}(y))dy,
\end{equation}
we have
\begin{equation}      \begin{split}                                                                            \label{eq61}
\vert a^{k+1}_1 (x)-a^{k}_1 (x)\vert &\leq \frac{\hbar^{k}}{k!}\int_{x}^{x_{0}}
\vert V_0(y)^{\frac{-1}{2}}V_{2}(y)\vert \rho^{k}(y)
dy \\&=-\frac{\hbar^{k}}{k!}\int_{x}^{x_{0}}
\rho'(y) \rho^{k}(y) dy =
\frac{\rho^{k+1}(x)\hbar^{k}}{(k+1)!}.
\end{split}
\end{equation}
We would like to have an estimate for the function $\rho(x)$, therefore it suffices that we obtain an
estimate for $\int_{x_{0}}^{x}\vert
V_0(y)\vert^{\frac{-1}{2}}\vert V_{2}(y)\vert\, dy$.
As already noted above, for $x>x_0$
\begin{equation*}
\rho(x)\le  \int_{x_{0}}^\infty \vert
V_0(y)\vert^{\frac{-1}{2}}\vert V_{2}(y)\vert\, dy \le C(\mu) <\infty
\end{equation*}
Hence, the solution to the integral equation \eqref{a1} is given by
\begin{equation}
a_{1}(x)=\sum_{s=1}^{\infty}(a^{s}_1 (x)-a^{s-1}_1 (x))
\end{equation}
and satisfies $\sup_{x>x_0}\vert a_{1}(x)\vert \leq C(\mu)<\infty$ uniformly in $0<\hbar<1$.
To derive the estimate for $a'_{1}(x)$, we observe that \[\frac{1}{\hbar} e^{\frac{2}{h} (S(y)-S(x))}=\frac{1}{2}\, V_0^{\frac{-1}{2}}(y)\, \partial_{y} e^{\frac{2}{\hbar} (S(y)-S(x))}.\] Therefore, using this observation and integrating by parts in \eqref{a1derivative} yields
\begin{align*}
 a_1'(x) &= \frac{1}{2}\int_{x_{0}}^{x} \Big(\frac{V_0(x)}{V_0(y)}\Big)^{\frac{1}{2}} \partial_y [e^{\frac{2}{\hbar}(S(y)-S(x))}-1]\,V_0(y)^{\frac{-1}{2}} V_{2}(y)\,dy \\
&\quad + \int_{x_{0}}^{x}[e^{\frac{2}{\hbar}(S(y)-S(x))}-1] \Big(\frac{V_0(x)}{V_0(y)}\Big)^{\frac{1}{2}} V_{2}(y) a_{1}(y)\,dy \\
& =-\frac{1}{2}\int_{x_{0}}^{x} V_0(x)^{\frac{1}{2}}  [e^{\frac{2}{\hbar}(S(y)-S(x))}-1]\, \partial_y [V_0(y)^{-1} V_{2}(y)]\,dy\\
&\quad - \frac{1}{2}[e^{\frac{2}{\hbar}(S(x_0)-S(x))}-1] \Big(\frac{V_0(x)}{V_0(x_0)}\Big)^{\frac{1}{2}} V_{2}(x_0) \\ & \quad + \int_{x_{0}}^{x}[e^{\frac{2}{\hbar}(S(y)-S(x))}-1] \Big(\frac{V_0(x)}{V_0(y)}\Big)^{\frac{1}{2}} V_{2}(y) a_{1}(y)\,dy.
\end{align*}
At this point we note that for $x>x_0$ ($x_0$ large enough), we have $\vert \partial_{x} ^{\ell} V_0(x, \hbar)\vert \leq c_{\ell,\,\mu} x^{-2-\ell}$, uniformly in $\hbar$.
Hence
\begin{equation}
\vert \partial_y [V_0(y)^{-1} V_{2}(y)]\vert \lesssim y^{-2},
\end{equation}
and using the boundedness of $a_{1}$, together with  that of $\rho(x)$, obviously implies that for $x>x_0$
\[
|a_1'(x)| \le C_{\mu} x^{-1},
\]
uniformly in $0<\hbar<1$ as desired.  For the case of $a_2$ one
proceeds in essentially the same way using, however,  the forward Green function rather than the
backward one. This yields
 \[
  a_{2} (x) =
\frac{1}{2}\int_{x}^\infty V_0(y)^{\frac{-1}{2}}[e^{\frac{-2}{\hbar}(S(y)-S(x))}-1]V_{2}(y)
[1+\hbar a_{2}(y)]\,dy.
 \]
The same arguments as before now show that $a_2$ satisfies~\eqref{eq3}.
The Wronskian $W(\psi_1, \psi_2)$ is obtain by evaluating at $x=\infty$.
\end{proof}

The same analysis of course yields zero energy solutions with the correct
asymptotic behavior as $x\to-\infty$. Note that the solution $\psi_2$ is the sub-ordinate one, i.e., it
is the unique (up to a nonzero scalar multiple) solution with the slowest possible growth. Hence, a zero energy
resonance in the context of Theorem~\ref{thm:main} would mean the existence of a nonzero solution to $P(x,\hbar D)f=0$
with $f(x)\sim c_{\pm} |x|^{\frac12-\alpha_{\pm}}$ as $x\to\pm\infty$ and with
\[
 \alpha_{\pm} = \sqrt{\frac14 + \mu_{\pm}^2\hbar^{-2}}
\]
It is easy to see that such a solution cannot exist if $V>0$. Indeed, let $\chi$ be a standard cut-off function with $\chi(0)=1$
and set $\chi_R(x):=\chi(x/R)$. If a globally subordinate solution $f(x)$  did exist, then
\[
 0= \limsup_{R\to\infty}\la P(\cdot,\hbar D) f,\chi_R f\ra  = \int [(f')^2(x) + V(x)f^2(x) ]\, dx
\]
implies that $f=0$, which is a contradiction.

\section{The Liouville-Green transform for small energies}
\label{sec:Langer}

In this section, we consider the equation
\begin{equation}                                                                    \label{nonzeroenergy.eq}
-\hbar^{2} f''_{\pm}(x) + V(x)f_{\pm}(x)= Ef_{\pm}(x)
\end{equation}
where $V$ is as in Theorem~\ref{thm:main}.  As explained in the introduction, we will
 use the Liouville-Green transform to reduce~\eqref{nonzeroenergy.eq} to a perturbed Airy equation.
We begin with a statement of the formal aspects (i.e., not involving estimates) of this transform, cf.\ Chapter~6 in~\cite{Olver}
and Langer's papers~\cite{Lang1}--\cite{Lang3}.
Henceforth, $V, V_0$ are as in Theorem~\ref{thm:main}.  Throughout this section $x\ge0$.

\begin{lem}\label{lem:langer}
There exists $E_0=E_0(V)>0$ so that for all $0<E<E_0$ one has the following properties:
the equation $V_0(x;\hbar)-E=0$ has a unique (simple) solution on $x>0$
which we denote by $x_1=x_1(E;\hbar)$. With $Q_0:=V_{0}-E$
\begin{equation}
  \label{eq:zeta}
   \zeta=\zeta(x,E;\hbar):= \sign (x-x_1(E;\hbar)) \Big|\frac32 \int_{x_1(E;\hbar)}^x
\sqrt{|Q_0(u,E;\hbar)|}\, du\Big|^{\frac23}
\end{equation}
defines a smooth change of variables $x\mapsto\zeta$
for all $x\ge0$. Let $q:=-\frac{Q_0}{\zeta}$. Then  $q>0$, $\frac{d\zeta}{dx}=\zeta'=\sqrt{q}$, and
\[
-\hbar^2f''+(V-E)f= 0
\]
transforms into
\begin{equation}\label{eq:Airy}
-\hbar^2 \ddot w(\zeta) = (\zeta+\hbar^2 \tilde V(\zeta,E;\hbar))w(\zeta)
\end{equation}
under $w= \sqrt{\zeta'} f = q^{\frac14} f$. Here $\dot{\
}=\frac{d}{d\zeta}$ and
\[
\tilde V := \frac{1}{4} q^{-1} \langle x\rangle ^{-2} - q^{-\frac14} \frac{d^2
q^{\frac14}}{d\zeta^2}
\]
\end{lem}
\begin{proof}
  Let $E_0>0$ be such that $V_0(x;\hbar)=E$ has a unique pair of solutions denoted by
$x_2(E;\hbar)<0<x_1(E;\hbar)$.   It is clear that~\eqref{eq:zeta} defines a smooth map away from $x=x_1(E;\hbar)$.
Taylor-expanding $Q_0(x,E;\hbar)$ in a neighborhood of that point and using that $V_0'(x_1(E;\hbar))<0$ implies
that $\zeta(x,E;\hbar)$ is smooth around $x=x_1$ as well with $\zeta'(x_1,E;\hbar)>0$.  Next, one checks that
\begin{align*}
 \dot w &= q^{-\frac14} f' + \frac{d q^{\frac14}}{d\zeta} f,\quad
 \ddot w = q^{-\frac34} f'' + \frac{d^2 q^{\frac14}}{d\zeta^2} f
\end{align*}
and thus, using $-\hbar^2f'' = (E-V)f$,
\begin{align*}
 -\hbar^2 \ddot w &= q^{-1} (E-V) w - \hbar^2 q^{-\frac14} \frac{d^2 q^{\frac14}}{d\zeta^2} w \\
&=  q^{-1} (-Q_0 + \hbar^2 \la x\ra^{-2} /4) w - \hbar^2 q^{-\frac14} \frac{d^2 q^{\frac14}}{d\zeta^2} w \\
&= \zeta w(\zeta) + \hbar^2 \big(q^{-1}\la x\ra^{-2}/4 - q^{-\frac14} \frac{d^2 q^{\frac14}}{d\zeta^2}\big) w
\end{align*}
as claimed.
\end{proof}

We now analyze the  properties of the change of variables introduced in the previous lemma.
Recall that
\[
 V_0(x;\hbar)= (\mu_+^2+\hbar^2/4)x^{-2}(1+O(x^{-1}))
\]
which implies that
\begin{equation}
 x_1(E;\hbar)= c(\hbar) E^{-\frac12} (1+O(E^{\frac12})),\quad c(\hbar)=\sqrt{\mu_+^2+\hbar^2/4}
\end{equation}
It will be convenient for us to normalize the constants here so that $c(\hbar)=\sqrt{2}$ and we shall assume
that for the remainder of this section. Moreover, we shall mostly suppress the harmless $\hbar$ dependence of various
functions in our notation.
 We begin with the following ``normal form'' lemma which will allow us to describe the function $\zeta$
in any region of the form $x\ge \eps E^{-\frac12}$ (which, in particular, contains the turning point). Note that Lemma~\ref{lem:eta} normalizes the
turning point $x_1$ to  $\xi=1$ by scaling out the energy.

\begin{lem}
  \label{lem:eta} Let $\eps>0$ but fixed. There exists a smooth map $\xi=\xi(y,E)$ on
  $(y,E)\in(\eps,\infty)\times (0,E_0)$ with $E_0$ small so that $\xi(E^{\frac12} x_1(E),E)=1$ and  for
  all $(y,E)$ in this range,
  \begin{equation}
    \label{eq:xy_corr} 1 - E^{-1} V_0(E^{-\frac12} y) =
    \Big(\frac{d\xi}{dy}\Big)^2 (1-\xi^{-2})
  \end{equation}
and such that, with some constant $\xi_0(E)$,
\begin{align}
\xi(y,E) &=  y + \xi_0(E) + y^{-1}\rho_0(y,E)
\label{eq:xforlargey}\\
|\partial_E^k\partial_y^{\ell} \xi(y,E)| &\le C_{k,\ell}\, E^{-k}
y^{1-\ell} \label{eq:xdiffinyE}
\end{align}
for all $k,\ell\ge0$. The functions $\xi_0$ and $\rho_0$
from~\eqref{eq:xforlargey} satisfy
\begin{equation}\label{eq:xi0rho0}
|\partial_E^k \xi_0(E)|\le C_k\, E^{-k},\qquad  |\partial_E^k
\partial_y^j \rho_0(y,E)|\le C_{jk}\, y^{-j}E^{-k}
\end{equation}
for all $k,j\ge0$ and uniformly in $(y,E)\in(1,\infty)\times
(0,E_0)$. For fixed $0<E<E_0$ the map $y\mapsto \xi(y,E)$ is a global
diffeomorphism whose inverse $y=y(\xi,E)$  satisfies the bounds
\begin{align}
y(\xi,E) &=  \xi + y_0(E) + \xi^{-1}\tilde\rho_0(\xi,E)  \label{eq:yforlargex}\\
|\partial_E^k\partial_\xi^{\ell} y(\xi,E)| &\le C_{k,\ell}\, E^{-k}
\xi^{1-\ell} \label{eq:ydiffinxE}
\end{align}
 for all $k,\ell\ge0$ and with functions $y_0$,
$\tilde\rho_0$ satisfying~\eqref{eq:xi0rho0} but relative to $\xi$ rather than~$y$.
All constants are allowed to depend on $\eps>0$.
\end{lem}
\begin{proof}
   Set $y_1=y_1(E):= E^{\frac12} x_1(E)$. Then $y_1=\sqrt{2}+O(E^{\frac12})$
   as $E\to0+$.
Note that the $O(\cdot)$ term here satisfies \[ \partial_E^k\, O(E^{\frac12}) = O(E^{\frac12-k})\quad \forall\;k\ge0\]
due to the corresponding assumption on the error term of~$V$. The same comment applies to every
$O(\cdot)$ term appearing in this proof, both with respect to derivatives in $E$ and spatial variables\footnote{We will say
that a $O(\cdot)$ term behaves like a symbol if its derivatives are governed by such power-laws.}.
Define the change of
   variables $\xi=\xi(y,E)$ via
   \begin{align}
     \int_{y_1}^y \sqrt{1-E^{-1} V_0(E^{-\frac12} u)}\; du &=
     \int_{1}^\xi \sqrt{1-t^{-2}}\, dt \qquad y>y_1 \label{eq:xdef1}\\
\int_{y}^{y_1} \sqrt{E^{-1} V_0(E^{-\frac12} u)-1}\; du &=
     \int_\xi^1 \sqrt{t^{-2}-1}\, dt \qquad \eps <y<y_1 \label{eq:xdef2}
   \end{align}
   By monotonicity, these identities define a unique correspondence
   between $\xi$ and~$y$ on these ranges which is, moreover, smooth
   and strictly increasing on $\eps<y<y_1$ and $y_1<y<\infty$.
   By inspection, they also satisfy~\eqref{eq:xy_corr}.
Since
\[
1-E^{-1} V_0(E^{-\frac12}
u)=1-\frac{2}{u^2} + O\big(E^{\frac12}{u}^{-3}\big)
\]
it follows furthermore that the interval $\eps \le y<\infty$ is
transformed into one of the form $0<\xi_1(E)<\xi<\infty$ where
\begin{equation}
  \label{eq:x1E} \xi_1(E)=\xi_1(0)+O(E^{\frac12})\text{\ \ as\ \ }E\to0+
\end{equation}
and $\xi_1(0)>0$ is a constant.  We first show that the map $\xi=\xi(y,E)$
so defined, is smooth for all $(y,E)\in (\eps,2)\times (0,E_0)$
together with the desired estimates. To this end, write
\begin{align*}
1-E^{-1} V_0(E^{-\frac12} y) &= (y-y_1) U(y,E) \\
U(y,E)&:= -E^{-\frac32}\int_0^1
V_0'\big(E^{-\frac12}(y_1+t(y-y_1))\big)\, dt
\end{align*}
Then for all $0<E\le E_0$ and all $k,\ell\ge0$, \begin{equation}
\label{eq:Ubd} \max_{1\le y\le2}|\partial_E^k \partial_y^\ell U(y,E)
|\le C_{k,\ell}\, E^{-k},\qquad \min_{1\le y\le2} U(y,E)\ge c_0>0
\end{equation}
For all $\eps<y<2$ we rewrite \eqref{eq:xdef1} and~\eqref{eq:xdef2} in
the form
\begin{equation}\label{eq:xy_corr2}
(y-y_1) Y(y,E) = (\xi-1) X(\xi)
\end{equation}
where
\begin{align*}
Y(y,E) &:=\Big(\int_0^1 \sqrt{(1-t)U(y_1+t(y-y_1),E)}\,
dt\Big)^{\frac23} \\ X(\xi) &:=\Big( \int_0^1
\frac{\sqrt{s(2+s(\xi-1))}}{1+s(\xi-1)}\, ds \Big)^{\frac23}
\end{align*}
By the preceding, \[ |\partial_E^k\partial_y^\ell Y(y,E)|\le
C_{k,\ell}E^{-k},\qquad |\partial_\xi^j\,X(\xi)|\le C_j  \qquad\forall\;k,\ell,j\ge0
\]
uniformly on the interval $\eps \le y\le2$ and the corresponding
interval in~$\xi$.  By the inverse function theorem,
\eqref{eq:xy_corr2} defines a (unique) smooth map also locally
around $\xi=1$ and $y=y_1$; this agrees with the previous definition
for $y\ne y_1$ and thus furnishes the desired smooth extension
through the point~$y=y_1$. Furthermore, from
\[
y_1= \sqrt{2}+O(E^{\frac12}),\qquad \partial_E^k y_1 =
O(E^{\frac12-k})
\]
and \eqref{eq:xy_corr2}, \eqref{eq:Ubd} we conclude  that
\[
\max_{1\le y\le2}|\partial_E^k \partial_y^\ell\, \xi(y,E) |\le
C_{k,\ell}\, E^{-k}
\]
for all $k,\ell\ge0$, $0<E<E_0$. For large $y$, we write
\[
\int_{y_1}^y \Big\{1-\frac{2}{u^2}\Big(1+O\big(
E^{\frac12}{u}^{-1}\big)\Big)\Big\}^{\frac12}\;du = \int_1^\xi
\sqrt{1-v^{-2}}\;dv
\]
The integral on the right-hand side satisfies
\[
\int_1^\xi \sqrt{1-v^{-2}}\;dv = \xi + \kappa + O(\xi^{-1})
\]
with a constant $\kappa$, whereas the one
on the left-hand side is equal to
\[
y+y_0 + O(y^{-1}) + O(E^{\frac12} y^{-2})
\]
with a constant $y_0(E)$. It is
easy to see that
\[
\xi + \kappa + O(\xi^{-1}) = y+y_0 + O(y^{-1}) + O(E^{\frac12} y^{-2})
\]
implies \eqref{eq:xforlargey} and we are done. The statements about
the inverse follow easily.
\end{proof}

We refer to Lemma~\ref{lem:eta} as a ``normal form''  since \eqref{nonzeroenergy.eq}, on the interval $x>E^{-\frac12}$,
turns into a suitably normalized (perturbed) Bessel equation in the variable~$\xi$, see the appendix.
By means of the change of variables introduced in Lemma~\ref{lem:eta} it is now
an easy matter to describe $\zeta(x,E)$ from Lemma~\ref{lem:langer} on the interval $x\ge E^{-\frac12}$.

\begin{cor}
 \label{cor:zeta} For all $0<E<E_0$ the following holds: there exists a constant $c_0>\sqrt{2}$ so that
on the interval $\eps <\sqrt{E}x<c_0$
\begin{equation}\label{eq:xinear1}
 \zeta(x,E) = 2^{\frac13}  (\xi-1)\big[1+ O(\xi-1) \big]
\end{equation}
with $O(\cdot)$ analytic and $\xi=\xi(\sqrt{E}x,E)$. For all $x\ge E^{-\frac12} c_0$,
\begin{equation}\label{eq:xibigger1}
 \frac23\zeta^{\frac23}(x,E) = \xi + \gamma + O(\xi^{-1})
\end{equation}
where $\gamma$ is some constant and $O(\cdot)$ is analytic.  Neither of the $O(\cdot)$ terms here
depend on~$E$ (other than through $\xi$).
\end{cor}
\begin{proof} We begin with $\xi$ close to $\xi=1$.
 The action $S(x,E)$ then satisfies, with $\xi$ as in Lemma~\ref{lem:eta},
\begin{align*}
 S(x,E) &= \sign(x-x_1(E)) \int_{x_1(E)}^x \sqrt{|E-Q_0(u)|}\, du \\
&=  \sign(x-x_1(E)) \int_{\sqrt{E}x_1(E)}^{\sqrt{E}x} \sqrt{|1-E^{-1}V_0(E^{-\frac12}y)|}\, dy \\
& = \sign(\xi(\sqrt{E}x)-1)  \int_{1}^{\xi(\sqrt{E}x,E)} \sqrt{|1-\eta^{-2}|}\, d\eta \\
& = \sqrt{2}\frac23 \sign(\xi-1) |\xi-1|^{\frac32} (1+ O(\xi-1))
\end{align*}
where the $O(\cdot)$ term is analytic in $|\xi-1|<1$  and $\xi=\xi(\sqrt{E}x,E)$.
In terms of~$\zeta$ this means that
\[
 \zeta(x,E) = 2^{\frac13}  (\xi-1)\big[1+ O(\xi-1) \big]
\]
which is \eqref{eq:xinear1}.  The constant $c_0$ is chosen so that $1<\xi(c_0 E^{-\frac12},E)<2$ for all $0<E<E_0$.
Since $x_1(E)=\sqrt{2/E} + o(1)$ as $E\to0$ we see that $c_0>\sqrt{2}$.
As for \eqref{eq:xibigger1},
\begin{align*}
 S(x,E) &=  \int_{1}^{\xi(\sqrt{E}x,E)} \sqrt{1-\eta^{-2}}\, d\eta \\
& = \int_1^\xi (1+ O(\eta^{-2}))\, d\eta = \xi + \gamma + O(\xi^{-1})
\end{align*}
Since $\zeta = (\frac32 S)^{\frac23}$, we are done.
\end{proof}

In the region $0<x<\eps x_1(E)$ we have the following description of $\zeta(x,E)$ with $\eps$
the same as in Lemma~\ref{lem:eta}. In fact, in the following lemma we will need $\eps$ small and
then use this choice in Lemma~\ref{lem:eta}.

\begin{lem}
 \label{lem:zetasmallx} For sufficiently small and fixed $\eps>0$  there exists a smooth function  $\tilde x(x,E)$
on $0\le x\le \eps x_1(E)$ with
\[
 \tilde x(x,E) = x(1+ O(Ex^2))
\]
and such that
\begin{equation}\label{eq:zetasmallx}
 \frac23 \zeta^{\frac32}(x,E) = \int_{\tilde x(x,E)}^{x_1(E)} \sqrt{V_0(v)}\, dv + O(E\log E)
\end{equation}
for all $0\le x\le \eps x_1(E)$ and $0<E<E_0$.  The $O(\cdot)$ here behave like symbols.
\end{lem}
\begin{proof} Define $\tilde x$ via
\[\begin{split}
 \int_0^{\tilde x(x,E)} \sqrt{V_0(v)}\, dv &= \int_0^x \sqrt{V_0(u)-E}\,du = \int_0^x \sqrt{V_0(u)}(1+O(E u^2))\,du \\
&= \int_0^x \sqrt{V_0(u)}\,du + O(E x^2)
\end{split}\]
Provided $\eps>0$ is sufficiently small (independently of $E$, of course), it follows from monotonicity considerations
that $\tilde x(x,E)$ exists with the desired properties. Next, note that for all $0<E<E_0$,
\[
 \int_0^{x_1(E)}\sqrt{V_0(u)}\, du - \int_0^{x_1(E)}\sqrt{V_0(u)-E}\, du = O(E\log E)
\]
and thus
\[\begin{split}
 \frac23 \zeta^{\frac32} &= \int_0^{x_1(E)}\sqrt{V_0(u)-E}\, du - \int_0^{x}\sqrt{V_0(u)-E}\, du \\
&=  \int_0^{x_1(E)}\sqrt{V_0(u)}\, du - \int_0^{\tilde x(x,E)} \sqrt{V_0(v)}\, dv + O(E\log E) \\
&= \int_{\tilde x(x,E)}^{x_1(E)} \sqrt{V_0(v)}\, dv + O(E\log E)
\end{split}
\]
as claimed.
\end{proof}

The point of~\eqref{eq:zetasmallx} is that $O(E\log E)$ is negligible as compared to the integral
on the right-hand side which is on the order of $|\log (E \la \tilde x\ra^2)|\gtrsim 1$. Thus, $\zeta$ behaves
to leading order like $|\log (E \la x\ra^2)|^{\frac23}$.  We now turn to estimating the functions
$q,\tilde V$ from Lemma~\ref{lem:langer}. In what follows, the notation $A\sim B$ will denote proportionality of
$A,B>0$ by some constants that are only allowed to depend on~$V$.
Also, $A\lesssim B$ will denote $A\le CB$ where $C$ is a constant, and similarly for $A\gtrsim B$.

\begin{lem}
  \label{lem:zeta2} Using the notations of Lemma~\ref{lem:langer}, let $0<E<E_0$. Then
on the interval $\zeta\ge-1$ the functions $q=q(\zeta,E)$ and $\tilde V=\tilde V(\zeta,E)$ satisfy
\begin{equation}\begin{aligned}
|\partial_E^k\partial_\zeta^\ell q| &\le C_{k,\ell}\,
E^{1-k}\la \zeta\ra^{-1-\ell} \\
|\partial_E^k\partial_\zeta^\ell
\tilde V(\zeta,E)| &\le C_{k,\ell}\, E^{-k} \la \zeta\ra^{-2-\ell}\qquad\forall\,k,\ell\ge0
\end{aligned}\label{eq:zetagr1}
\end{equation}
On
the interval $\zeta(0,E)\le\zeta\le -1$ we view $q,\tilde V$ as functions of~$x$ via~\eqref{eq:zeta}. Then one has $q\sim
|\zeta|^{-1}\,\la x\ra^{-2}$ and there is the representation
\begin{equation}\label{eq:vtilde_rep}
 \tilde V(\zeta,E)  = -\frac{5}{16\zeta^2}  +  q^{-1} (E\,
\beta_0(x,E) + \la x\ra^{-3} \beta_1(x,E))
\end{equation}
where $\beta_j$ satisfy the bounds
\begin{equation}\label{eq:zetaklmin1}\begin{split} |
\partial_E^k\partial_x^\ell \beta_j(x,E)| &\le C_{k,\ell}\, E^{-k}\la
 x\ra^{-\ell} \qquad j=0,1\\
| \partial_E^k\partial_x^\ell  q| &\le C_{k,\ell}\,E^{-k}
|\zeta|^{-1}\,\la x\ra^{-2-\ell} \qquad\forall\,k,\ell\ge0
\end{split}\end{equation}
All constants are  independent of $E$.
\end{lem}
\begin{proof} The case $\zeta\ge-1$ corresponds to $x\ge \eps x_1(E)$ by Lemma~\ref{lem:eta} and Corollary~\ref{cor:zeta}.
We now use that corollary to write
\[
 \zeta=\zeta(\xi,E), \quad \xi=\xi(y,E),\quad y=E^{\frac12} x
\]
Then
\[ q= (\zeta')^2 = E (\partial_\xi \zeta(\xi,E))^2 (\partial_y \xi(y,E))^2 \sim E
 \]
The derivative bounds on $q$ now follow from those obtained in Lemma~\ref{lem:eta} and Corollary~\ref{cor:zeta}.
As for $\tilde V$, we compute, with $\dot{\;}=\frac{d}{d\zeta}$,
\begin{equation}\label{eq:tilV}
\tilde V= \frac{1}{4}q^{-1} \langle x\rangle ^{-2} -q^{-\frac14} \frac{d^2
q^{\frac14}}{d^2\zeta} = \frac{1}{4}q^{-1} \langle x\rangle ^{-2}+\frac{3}{16} q^{-2}
\dot q^2 - \frac14 q^{-1} \ddot q
\end{equation}
From the bounds on $q$ which we just derived, the last two terms on the right-hand side
of~\eqref{eq:tilV} are $\les \zeta^{-2}$ and behave as stated under differentiation. To treat the first
term, we invoke \eqref{eq:yforlargex}, \eqref{eq:xinear1}, and \eqref{eq:xibigger1} to write
\begin{align*}
 q^{-1}\la x\ra^{-2} &= q^{-1} \la E^{-\frac12} y\ra^{-2} = q^{-1} \la E^{-\frac12} y(\xi,E)\ra^{-2}\\
&= q^{-1} \la E^{-\frac12} y(\xi(\zeta,E),E)\ra^{-2}
\end{align*}
which implies the correct bounds. Indeed, if $|\zeta|\les 1$, then $q\sim E$ and the change of variables $\zeta\mapsto\xi\mapsto y$
has derivatives of size $\les 1$ relative to $\zeta$ uniformly in~$E$. This implies that $q^{-1}\la x\ra^{-2}\sim 1$ with
derivatives with respect to $\zeta$ of size $\les 1$; furthermore, each derivative in~$E$ costs one power of~$E$. Next, if $\zeta\ge1$,
then the change of variables $\zeta\mapsto\xi\mapsto y$ acts like $\zeta^{\frac23}$ by Corollary~\ref{cor:zeta}. Thus,
\[
 q^{-1}\la x\ra^{-2} \sim E^{-1}\zeta (E^{-\frac12} \zeta^{\frac32})^{-2} \sim \zeta^{-2}
\]
with each $\zeta$ derivative gaining one more power of decay in~$\zeta$.

\noindent In the remaining case $\zeta\le -1$ one first calculates, on the one
hand,
\begin{align*}
q^{-\frac14} \frac{d^2 q^{\frac14}}{d\zeta^2} &=
\frac{5}{16\zeta^2} +\frac{1}{4}\frac{\ddot Q_0}{Q_0} -
\frac18\frac{\dot Q_0}{\zeta
Q_0} -\frac{3}{16} \Big(\frac{\dot Q_0}{Q_0}\Big)^2 \\
& = \frac{5}{16\zeta^2} + q^{-1}\Big[ \frac14 \frac{V_0''}{Q_0}
-\frac{5}{16} \Big(\frac{V_0'}{Q_0}\Big)^2\Big]
\end{align*}
where ${\ }'=\frac{d}{dx}$. Thus, from \eqref{eq:tilV},
\begin{align*}
 \tilde V &= \frac{1}{4}q^{-1} \langle x\rangle ^{-2}- q^{-\frac14} \frac{d^2 q^{\frac14}}{d\zeta^2} \\
&= -\frac{5}{16\zeta^2} + q^{-1}\Big[\frac14\langle x\rangle ^{-2}- \frac14 \frac{V_0''}{Q_0}
+\frac{5}{16} \Big(\frac{V_0'}{Q_0}\Big)^2\Big] \\
&= -\frac{5}{16\zeta^2} + q^{-1}\Big[E\beta_0(x,E) + \la x\ra^{-3}\beta_1(x,E)\Big]
\end{align*}
where we have set\footnote{$\beta_1$ does not depend on $E$, but this makes no difference.}
\begin{align}
\beta_0(x,E) &:= E^{-1} \Big[ \frac14\Big(\frac{V_0''}{V_0} -\frac{V_0''}{Q_0}\Big)
+\frac{5}{16} \Big(\Big(\frac{V_0'}{Q_0}\Big)^2  - \Big(\frac{V_0'}{V_0}\Big)^2 \Big) \Big] \label{eq:beta0def}\\
 \beta_1(x,E) &:= \la x\ra^{3}\Big[\frac14\langle x\rangle ^{-2}- \frac14 \frac{V_0''}{V_0}
+\frac{5}{16} \Big(\frac{V_0'}{V_0}\Big)^2 \Big] \label{eq:beta1def}
  \end{align}
As already noted in Section~\ref{sec:zero}, the $x^{-2}$ terms inside the brackets in~\eqref{eq:beta1def}
cancel so that the leading order is $x^{-3}$. In fact,
$|\partial_x^\ell \beta_1(x,E)|\le C_\ell \la x\ra^{-\ell}$ in view of our assumptions on~$V$, see
Theorem~\ref{thm:main}. As for $\beta_0$, we note that in the range $\zeta\le -1$, one has $Q_0\sim V_0$.
Since $Q_0=V_0-E$, this implies that the expression in brackets in~\eqref{eq:beta0def} is $\les E$ together with the natural
derivative bounds.
The bounds on
\[q=(V_0-E)|\zeta|^{-1} \sim \la x\ra^{-2} |\zeta|^{-1},\qquad |\zeta|\sim |\log(E\la x\ra^2)|^{\frac23}
\]
follow from~\eqref{eq:zetasmallx} and we are done.
\end{proof}

In Lemma~\ref{lem:zeta2} the modification of $V$ to $V_0$ only played a role in the regime $\zeta\ll-1$
which is the same as $x< \eps x_1(E)$. This is natural, since we know from Section~\ref{sec:zero} that this modification
really comes from the $E=0$ case which corresponds to $x_1=+\infty$. We will see this mechanism at work in the following
section, too.

\section{Solving the perturbed Airy equation}\label{sec:Airy}

This section is devoted to solving \eqref{eq:Airy}, at least in the asymptotic sense relative to~$\hbar$.
We shall use the notations and results of the previous section. For the properties of the
Airy functions $\Ai, \Bi$ listed below we refer the reader to Chapter~11 of~\cite{Olver}.

\begin{prop}
  \label{prop:AiryI} Let $\hbar_0>0$ be small. A fundamental system of solutions to~\eqref{eq:Airy} in the range
$\zeta\le0$ is given by
\begin{equation}\nonumber
\begin{aligned}
  \phi_1(\zeta,E,\hbar) &=
  \Ai(\tau) [1+\hbar a_1(\zeta,E,\hbar)] \\
  \phi_2(\zeta,E,\hbar) &=
  \Bi(\tau) [1+\hbar a_2(\zeta,E,\hbar)]
\end{aligned}
\end{equation}
 with
$\tau:=-\hbar^{-\frac23}\zeta$.
 Here $a_1, a_2$ are smooth, real-valued, and they satisfy
the bounds, for all $k\ge0$ and $j=1,2$, and with $\zeta_0:=\zeta(0,E)$,
\begin{equation}\label{eq:aj_est}
\begin{aligned}
  |\partial_E^k  a_j(\zeta,E,\hbar)| &\less  E^{-k} \min\big[\hbar^{\frac13} \la
  \hbar^{-\frac23}
  \zeta\ra^{\frac12}, 1\big] \\
 |\partial_E^k \partial_\zeta a_j(\zeta,E,\hbar)| &\les E^{-k} \Big[\hbar^{-\frac13} \la
\hbar^{-\frac23}\zeta\ra^{-\frac12}\chi_{[-1\le \zeta\le0]} +
|\zeta|^{\frac12}  \chi_{[ \zeta_0\le \zeta\le -1]}\Big ]
\end{aligned}
\end{equation}
uniformly in the parameters $0<\hbar<\hbar_0$, $0<E<E_0$.
\end{prop}
\begin{proof}
Let $\phi_{1,0}(\zeta,\hbar):=\Ai(\tau)$ and
$\phi_{2,0}(\zeta,\hbar):=\Bi(\tau)$. We seek a basis of the form
\[
\phi_j(\zeta)=\phi(\zeta,\hbar,E) = \phi_{j,0}(\zeta,\hbar)(1+ h
a_{j}(\zeta,\hbar,E))
\] for $\zeta\le0$
This representation is meaningless for $\zeta>0$ since $\phi_{j,0}$
have real zeros there. On the other hand, on $\zeta\le 0$ they do
not vanish.  We obtain the
equation
\begin{equation}
  \label{eq:mult_ODE} (\phi_{j,0}^2 \dot a_j)^{\dot{}} = - \frac{1}{\hbar } \tilde V \phi_{j,0}^2
(1+ \hbar a_{j})
\end{equation} for $j=1,2$ where $\dot{\;}=\partial_\zeta$. A solution of~\eqref{eq:mult_ODE} on
$\zeta\le0$ is given by, with $a_2(\zeta)=a_2(\zeta,\hbar,E)$,
\begin{align}
  a_2(\zeta) &:= -\frac{1}{\hbar } \int_\zeta^0
\phi_{2,0}^2(\eta,\hbar) \int_{\zeta}^{\eta}
\phi_{2,0}^{-2}(\tilde\eta,\hbar)\,d\tilde\eta\,
\tilde V(\eta,E)
(1+ \hbar a_{2}(\eta))\, d\eta \nonumber \\
&= -\hbar^{\frac13}\int_{0}^{-\hbar^{-\frac23}\zeta} \Bi^2(u)
\Big[\int_u^{-\hbar^{-\frac23}\zeta} \Bi^{-2}(v)\, dv\Big]
\tilde V(-\hbar^{\frac23}u,E) (1+ \hbar a_{2}(\hbar^{\frac23}u))\, du
\label{eq:a2}
\end{align}
This solution is unique with the property that $a_2(0)=\dot a_2(0)=0$.
Recall the asymptotic behavior, see~\cite{Olver},
\[\begin{split}
\Bi(x) &= \pi^{-\frac12} x^{-\frac14} e^{\frac23 x^{\frac32}}
\big[1+O(x^{-\frac32})\big] \text{\ \ as\ \ } x\to\infty
\\\Bi(x)&\ge\Bi(0) >0\quad \forall\;x\ge0 \\
\Ai(x) &= \frac12\pi^{-\frac12} x^{-\frac14} e^{-\frac23
x^{\frac32}} \big[1+O( x^{-\frac32})\big] \text{\ \ as\ \ }
x\to\infty
\\\Ai(x)& >0\quad \forall\;x\ge0
\end{split}\]
Also note the useful fact, valid for any $0\le x_0<x_1$,
\begin{equation}\label{eq:AiBiint}
\int_{x_0}^{x_1} \Bi^{-2}(y)\, dy = \pi^{-1}
\Big(\frac{\Ai(x_0)}{\Bi(x_0)} - \frac{\Ai(x_1)}{\Bi(x_1)}\Big)
\end{equation}
which implies that
\[
\Big| \Bi^2(x_0) \int_{x_0}^{x_1} \Bi^{-2}(y)\, dy  \Big| \les \la
x_0\ra^{-\frac12}
\]
The leading term in~\eqref{eq:a2}, i.e.,
\[
a_{2,0}(\zeta,E,\hbar) :=
-\hbar^{\frac13}\int_{0}^{-\hbar^{-\frac23}\zeta} \Bi^2(u)
\Big[\int_u^{-\hbar^{-\frac23}\zeta} \Bi^{-2}(v)\, dv\Big]\,
\tilde V(-\hbar^{\frac23}u,E) \, du
\]
therefore satisfies the bound (dropping $E,\hbar$ from $a_{2,0}$ for simplicity)
\[
|a_{2,0}(\zeta)| \les \hbar^{\frac13}\int_{0}^{-\hbar^{-\frac23}\zeta}
\la u\ra^{-\frac12}  |\tilde V(-\hbar^{\frac23}u,E)| \, du
\]
We now use the estimates from  Lemma~\ref{lem:zeta2} to bound the right-hand side.
If $-1\le\zeta\le0$, then this yields
\begin{equation}\label{eq:a20_bd}
|a_{2,0}(\zeta)|\lesssim  \hbar^{\frac13}
\int_{0}^{-\hbar^{-\frac23}\zeta} \la u\ra^{-\frac12}\, du
\lesssim \hbar^{\frac13} \la \hbar^{-\frac23}\zeta\ra^{\frac12}
\end{equation} 
On the other hand, if $\zeta_0:=\zeta(0,E)\le\zeta\le -1$, then we obtain
\begin{align}
|a_{2,0}(\zeta)| &\lesssim  \hbar^{\frac13}
\int_{0}^{\hbar^{-\frac23}} \la u\ra^{-\frac12}\, du \label{eq:parta}\\
&\quad +
\hbar^{\frac13} \int_{\hbar^{-\frac23}}^{-\hbar^{-\frac23}\zeta}
u^{-\frac12}\Big[(\hbar^{\frac23} u)^{-2} + q^{-1}(-\hbar^{\frac23}u)
(E+ \la z\ra^{-3})\Big]\, du \label{eq:partb}
\end{align}
The variable $z$ appearing in~\eqref{eq:partb} is tied to the integration variable $u$ via $-\hbar^{_\frac23}u=\zeta(z,E)$,
see Lemma~\ref{lem:langer}.
The integral in~\eqref{eq:parta} and the first term inside the brackets in~\eqref{eq:partb} contribute
\[
 \hbar^{\frac13}
\int_{0}^{\hbar^{-\frac23}} \la u\ra^{-\frac12}\, du + \hbar^{\frac13} \int_{\hbar^{-\frac23}}^{-\hbar^{-\frac23}\zeta}
u^{-\frac12} (\hbar^{\frac23} u)^{-2}\, du \les 1
\]
Next, with $\zeta(x_2,E)=-1$, and $v=\zeta(z,E)$, $dv=\sqrt{q}dz$,
\begin{align*}
 \hbar^{\frac13} \int_{\hbar^{-\frac23}}^{-\hbar^{-\frac23}\zeta}
u^{-\frac12} q^{-1}(\hbar^{\frac23}u) E\, du &= E \int_1^{-\zeta}
v^{-\frac12} q^{-1}(v) \, dv \\
& = E \int_{x}^{x_2}
Q_0^{-\frac12}(z,E)\, dz
\les E \int_{x}^{x_1} z \, dz \\ &\les E x_1(E)^2 \les 1
\end{align*}
Finally,  using that $\frac{dv}{dz}=\sqrt{q}$ once again one obtains
\begin{align*}
 \hbar^{\frac13} \int_{\hbar^{-\frac23}}^{-\hbar^{-\frac23}\zeta}
u^{-\frac12} q^{-1}(\hbar^{\frac23}u) z^{-3}\, du &=  \int_1^{-\zeta}
v^{-\frac12} q^{-1}(v) \la z\ra^{-3}\, dv \\
& = \int_{x}^{x_2}
Q_0^{-\frac12}(z,E)\la z\ra^{-3}\, dz \\
&\les  \int_{x}^{x_1} \la z\ra^{-2} \, dz \les  1
\end{align*}
In summary\footnote{Had we used $V$ instead of $V_0$ in our definition of~$\zeta$, then  we would be
losing a factor of $\log E$ at this point. Indeed, for the case of  $V$ we would need to replace $E+\la z\ra^{-3}$
 by the strictly weaker $\la z\ra^{-2}$
in~\eqref{eq:partb} which then leads to the logarithmically divergent integral $\int_{x}^{x_1} \la z\ra^{-1} \, dz$.},
\[
|a_{2,0}(\zeta,E,\hbar)| \les \min(1, \hbar^{\frac13} \la
\hbar^{-\frac23}\zeta\ra^{\frac12})
\]
uniformly in $\zeta\in[\zeta_0,0]$, $0<E<E_0$, and $0<\hbar<\hbar_0$. Due to
the linear nature of~\eqref{eq:a2},  a contraction argument now
yields the same bound for $a_2$; in fact, due to the derivative
bounds of Lemma~\ref{lem:zeta2}, we obtain the more general estimate
\[
|\partial_E^k a_{2}(\zeta,E,\hbar)| \le C_k\, E^{-k} \min(1, \hbar^{\frac13}
\la \hbar^{-\frac23}\zeta\ra^{\frac12})\quad \forall \;k\ge0
\]
uniformly in the parameters. As for the first derivative in $\zeta$,
observe that
\begin{equation}\label{eq:aminprim}
\dot a_2(\zeta) = \frac{\hbar^{-1}}{
\phi_{2,0}^{2}(\zeta,\hbar)} \int_\zeta^0
\phi_{2,0}^2(\eta,\hbar)
 \tilde
V(\eta,E) (1+ \hbar a_{2}(\eta))\, d\eta
\end{equation}
whence, for all $-1\le\zeta\le0$,
\begin{align*}
|\dot a_2(\zeta,E,\hbar)| & \les \hbar^{-1}\,
\phi_{2,0}^{-2}(\zeta,\hbar) \int_\zeta^0 \phi_{2,0}^2(\eta,\hbar)
|\tilde
V(\eta,E)|\, d\eta  \\
&\les \hbar^{-\frac13} \Bi^{-2}(-\hbar^{-\frac23}\zeta) \int_{0}^{-\hbar^{-\frac23}\zeta} \Bi^2(u)\, du \\
&\les \hbar^{-\frac13} \la \hbar^{-\frac23}\zeta\ra^{\frac12} e^{-\frac43 \hbar^{-1} |\zeta|^{\frac32}} \int_{0}^{-\hbar^{-\frac23}\zeta} \la u\ra^{-\frac12} e^{\frac43 u^{\frac32}}\, du\\
&\les  \hbar^{-\frac13} \la \hbar^{-\frac23}\zeta\ra^{\frac12} e^{-\frac43 \hbar^{-1} |\zeta|^{\frac32}} \int_{0}^{\hbar^{-1}|\zeta|^{\frac32}} \la v\ra^{-\frac13} |v|^{-\frac13}e^{\frac43 v}\, dv\\
&\les \hbar^{-\frac13} \la \hbar^{-\frac23}\zeta\ra^{\frac12} e^{-\frac43 \hbar^{-1} |\zeta|^{\frac32}} \la \hbar^{-1}|\zeta|^{\frac32}\ra^{-\frac23} e^{\frac43 \hbar^{-1}|\zeta|^{\frac32}}\\
&\les \hbar^{-\frac13} \la \hbar^{-\frac23}\zeta\ra^{-\frac12}
\end{align*}
If
$\zeta_0\le \zeta\le -1$, then
\begin{align}
&|\dot a_2(\zeta,E,\hbar)|  \les \hbar^{-1}\,
\phi_{2,0}^{-2}(\zeta,\hbar) \int_\zeta^0 \phi_{2,0}^2(\eta,\hbar)
|\tilde
V(\eta,E)|\, d\eta  \nonumber \\
&\les \hbar^{-\frac13} \Bi^{-2}(-\hbar^{-\frac23}\zeta)
\int_{0}^{\hbar^{-\frac23}} \Bi^2(u)\, du \label{eq:teil1}\\
&\quad + \hbar^{-\frac13} \Bi^{-2}(-\hbar^{-\frac23}\zeta)
\int_{\hbar^{-\frac23}}^{-\hbar^{-\frac23}\zeta}\Bi^2(u)\Big[ (\hbar^{\frac23}
u)^{-2} +  \frac{E+\la z\ra^{-3}}{q(-\hbar^{\frac23}u)} \Big]\, du \label{eq:teil2}
\end{align}
where $z$ has the same meaning as in~\eqref{eq:partb}.
First, note that \eqref{eq:teil1} is rapidly (in fact, super
exponentially) decreasing as $|\zeta|$ increases: $
|\eqref{eq:teil1}|\less e^{-|\zeta|^{\frac32}}.$  Second, we bound
the first part of~\eqref{eq:teil2} by
\begin{align*}
& \hbar^{-\frac13} \Bi^{-2}(-\hbar^{-\frac23}\zeta)
\int_{\hbar^{-\frac23}}^{-\hbar^{-\frac23}\zeta}\Bi^2(u) (\hbar^{\frac23}
u)^{-2} \, du \nonumber \\
&\les \hbar^{-\frac53} \Bi^{-2}(-\hbar^{-\frac23}\zeta)
\int_{\hbar^{-\frac23}}^{-\hbar^{-\frac23}\zeta} u^{-\frac52} e^{\frac43
u^{\frac32}} \, du  \nonumber\\
&\les \hbar^{-\frac53} (\hbar^{-\frac23}|\zeta|)^{-\frac52} \les
|\zeta|^{-\frac52}
\end{align*}
The contribution to~\eqref{eq:teil2} involving $E$ is
\begin{align}
& \hbar^{-\frac13} \Bi^{-2}(-\hbar^{-\frac23}\zeta)
\int_{\hbar^{-\frac23}}^{-\hbar^{-\frac23}\zeta}\Bi^2(u)
(q(\hbar^{\frac23}u))^{-1} E \, du \nonumber\\
&\les E \hbar^{-\frac13} (\hbar^{-\frac23}|\zeta|)^{\frac12}
e^{-\frac{4}{3\hbar}|\zeta|^{\frac32}}
\int_{\hbar^{-\frac23}}^{-\hbar^{-\frac23}\zeta}
u^{-\frac12} e^{\frac43 u^{\frac32}} (q(\hbar^{\frac23} u))^{-1}\, du \label{eq:myst_form1} \\
&\les E\hbar^{-1} |\zeta|^{\frac12}
e^{-\frac{4}{3\hbar}|\zeta|^{\frac32}}\int_1^{-\zeta} z\,
e^{\frac{4}{3\hbar} v^{\frac32}}\, dz \les |\zeta|^{\frac12} E\la x\ra^2
\label{eq:myst_form2}
\end{align}
To pass from \eqref{eq:myst_form1} to \eqref{eq:myst_form2}, we
 substituted $u= \hbar^{-\frac23}v$ and then changed variables $dv =
 \sqrt{q(-v)}\, dz$ followed by $v q(-v)= Q_0(z)$; in~\eqref{eq:myst_form2} the relation
 between $v$ and~$z$, as well as $\zeta$ and $x$, is given by~\eqref{eq:zeta}. i..e, $v=\zeta(z,E)$, $\zeta=\zeta(x,E)$. To pass to the
 final inequality in~\eqref{eq:myst_form2} we integrate by parts so as to gain a factor of~$\hbar$:
\begin{align*}
 \int_1^{-\zeta} z\, e^{\frac{4}{3\hbar} v^{\frac32}}\, dz &\les
\int_x^{x_2} z\, e^{\frac{2}{\hbar } \int_{z}^{x_1}
\sqrt{Q_0(\eta)}\, d\eta} \, dz \\
&\les \hbar \la x\ra^2 e^{\frac{2}{\hbar } \int_{x}^{x_1}
\sqrt{Q_0(\eta)}\, d\eta} = \hbar\la x \ra^2
e^{\frac{4}{3\hbar}|\zeta|^{\frac32}}
\end{align*}
where $\zeta(x_2,\hbar)=-1$.   Finally, we turn the contribution of $\la z\ra^{-3}$ in~\eqref{eq:teil2}.
Using the same conventions regarding the relation between the variables this contribution is of the form
\begin{align*}
& \hbar^{-\frac13} \Bi^{-2}(-\hbar^{-\frac23}\zeta)
\int_{\hbar^{-\frac23}}^{-\hbar^{-\frac23}\zeta}\Bi^2(u)
(q(\hbar^{\frac23}u))^{-1} \la z\ra^{-3} \, du \nonumber\\
&\les  \hbar^{-\frac13} (\hbar^{-\frac23}|\zeta|)^{\frac12}
e^{-\frac{4}{3\hbar}|\zeta|^{\frac32}}
\int_{\hbar^{-\frac23}}^{-\hbar^{-\frac23}\zeta}
u^{-\frac12} e^{\frac43 u^{\frac32}} (q(\hbar^{\frac23} u))^{-1}\la z\ra^{-3} \, du  \\
&\les \hbar^{-1} |\zeta|^{\frac12}
e^{-\frac{4}{3\hbar}|\zeta|^{\frac32}}\int_1^{-\zeta}\la z\ra^{-2}\,
e^{\frac{4}{3\hbar} v^{\frac32}}\, dz \les |\zeta|^{\frac12} \la x\ra^{-2}
\end{align*}
The final inequality here is based on the same kind of integration by parts as before:
\begin{align*}
 \int_1^{-\zeta} \la z\ra^{-2}\, e^{\frac{4}{3\hbar} v^{\frac32}}\, dz &\les
\int_x^{x_2} \la z\ra^{-2}\, e^{\frac{2}{\hbar } \int_{z}^{x_1}
\sqrt{Q_0(\eta)}\, d\eta} \, dz \\
&\les \hbar \la x\ra^{-2} e^{\frac{2}{\hbar } \int_{x}^{x_1}
\sqrt{Q_0(\eta)}\, d\eta} = \hbar\la x \ra^{-2}
e^{\frac{4}{3\hbar}|\zeta|^{\frac32}}
\end{align*}
with $x_2$ as above.
In conclusion, we estimate the
contributions of~\eqref{eq:teil1} and~\eqref{eq:teil2} by
\begin{equation}\label{eq:erst_ableit}
|\dot a_2(\zeta,E,\hbar)| \les \hbar^{-\frac13} \la
\hbar^{-\frac23}\zeta\ra^{-\frac12}\chi_{[-1\le \zeta\le0]} +
|\zeta|^{\frac12}  \chi_{[ \zeta_0\le \zeta\le -1]}
\end{equation}
as claimed.

\noindent
Next, we turn to $\phi_1(\zeta,E)$ (dropping $\hbar$ for simplicity). As usual we make the
reduction ansatz
\[
\phi_1(\zeta,E) = g(\zeta,E) \phi_2(\zeta,E)
\]
which leads to the equation $(\phi_2^2 \dot g)^{\dot{\;}}=0$.  At this point
it is convenient to extend the solutions $\phi_2$, which are originally
defined on the interval $\zeta(0,E)\le \zeta\le0$, to all of $\zeta\le0$.
This is done in such a way that the bounds~\eqref{eq:aj_est} remain valid for $\zeta\le\zeta_0$ without,
however, making any reference to the ODE~\eqref{eq:Airy} for those~$\zeta$.
We can now solve for
for $g$ in the form
\[
 \phi_1(\zeta,E) = \pi\hbar^{-\frac23}\phi_2(\zeta,E)
 \int_{-\infty}^\zeta \phi_2(\eta,E)^{-2}\, d\eta
\]
 Inserting our representation
of~$\phi_2$ into this formula yields
\begin{align*}
\phi_1(\zeta,E) &=  \pi\hbar^{-\frac23}
  \Bi(-\hbar^{-\frac23}\zeta) [1+\hbar a_2(\zeta,E)]  \int_{-\infty}^\zeta
  \Bi^{-2}(-\hbar^{-\frac23}\eta) [1+\hbar a_2(\eta,E)]^{-2} \,
  d\eta
\end{align*}
First, we note that from~\eqref{eq:AiBiint},
\[
 \pi\hbar^{-\frac23}
  \Bi(-\hbar^{-\frac23}\zeta)  \int_{-\infty}^\zeta
  \Bi^{-2}(-\hbar^{-\frac23}\eta) \,
  d\eta = \Ai(-\hbar^{-\frac23}\zeta)
\]
Second, $[1+\hbar a_2]^{-2}=1+\hbar \tilde a_2$ where $\tilde a_2$ satisfies the same bounds as $a_2$
(since $|a_2|\les 1$).  Thus, inspection of our formula for $\phi_1$ reveals that $a_1=\pi(a_2 + \tilde a_1)$ where
\[
\begin{split}
 \tilde a_1(\zeta) &:= \hbar^{-\frac23}
  \frac{\Bi}{\Ai}(-\hbar^{-\frac23}\zeta) [1+\hbar a_2(\zeta,E)]  \int_{-\infty}^\zeta
  \Bi^{-2}(-\hbar^{-\frac23}\eta) \tilde a_2(\eta,E) \,
  d\eta \\
& =
  \frac{\Bi}{\Ai}(-\hbar^{-\frac23}\zeta) [1+\hbar a_2(\zeta,E)]  \int_{-\hbar^{-\frac23}\zeta}^{\infty}
  \Bi^{-2}(\eta) \tilde a_2(\hbar^{\frac23}\eta,E) \,
  d\eta
\end{split}
\]
Furthermore, from~\eqref{eq:AiBiint},
\begin{align}
 &\pi\int_{-\hbar^{-\frac23}\zeta}^{\infty}
  \Bi^{-2}(\eta) \tilde a_2(\hbar^{\frac23}\eta,E) \,
  d\eta = -\int_{-\hbar^{-\frac23}\zeta}^{\infty} \tilde a_2(-\hbar^{\frac23}\eta,E)\,d\Big[\frac{\Ai}{\Bi}(\eta)\Big] \nonumber\\
& = \frac{\Ai}{\Bi}(-\hbar^{-\frac23}\zeta) \tilde a_2(\zeta,E) - \hbar^{\frac23} \int_{-\hbar^{-\frac23}\zeta}^{\infty} \frac{\Ai}{\Bi}(\eta) (\partial_1 {\tilde a_2})(-\hbar^{\frac23}\eta,E)\,d\eta \label{eq:pain}
\end{align}
where $\partial_1$ refers to the derivative in the first variable.  The first term in~\eqref{eq:pain} makes an admissible
contribution to~$a_1$ whereas the second one is controlled as follows:
\begin{align*}
 & \hbar^{\frac23}\frac{\Bi}{\Ai}(-\hbar^{-\frac23}\zeta)  \int_{-\hbar^{-\frac23}\zeta}^{\infty} \frac{\Ai}{\Bi}(\eta) \big|(\partial_1 {\tilde a_2})(-\hbar^{\frac23}\eta,E)\big|\,d\eta \\
&\les \hbar^{\frac23} e^{\frac{4}{3\hbar}\la\zeta\ra^{\frac32}} \int_{-\hbar^{-\frac23}\zeta}^{\infty}  e^{-\frac{4}{3\hbar}\la\eta\ra^{\frac32}} \Big[\hbar^{-\frac13}\la\eta\ra^{-\frac12}\chi_{[-1\le\hbar^{\frac23}\eta\le0]} +
|\hbar^{\frac23}\eta|^{\frac12} \chi_{[\le\hbar^{\frac23}\eta\le-1]} \Big] \,d\eta \\
&\les \hbar^{\frac13}\la \hbar^{-\frac23}\zeta\ra^{-1}  \chi_{[-1\le \zeta\le0]} + \hbar \chi_{[\zeta_0\le\zeta\le-1]}
\les \hbar^{\frac13}\la \hbar^{-\frac23}\zeta\ra^{\frac12}  \chi_{[-1\le \zeta\le0]} + \chi_{[\zeta_0\le\zeta\le-1]}
\end{align*}
as desired.  For the derivative in $\zeta$,
\begin{align*}
 \partial_\zeta \tilde a_1(\zeta)  & =  -\pi\hbar^{-\frac23}
  \Ai^{-2}(-\hbar^{-\frac23}\zeta) [1+\hbar a_2(\zeta,E)]  \int_{-\hbar^{-\frac23}\zeta}^{\infty}
  \Bi^{-2}(\eta) \tilde a_2(\hbar^{\frac23}\eta,E) \,
  d\eta  \\
&\quad + \hbar^{-\frac23} (\Ai\Bi)^{-1}(-\hbar^{-\frac23}\zeta) [1+\hbar a_2(\zeta,E)]
\tilde a_2(\zeta,E) \\
& \quad + \hbar \frac{\Bi}{\Ai}(-\hbar^{-\frac23}\zeta) \dot a_2(\zeta,E)  \int_{-\hbar^{-\frac23}\zeta}^{\infty}
  \Bi^{-2}(\eta) \tilde a_2(\hbar^{\frac23}\eta,E) \,
  d\eta
\end{align*}
Using \eqref{eq:pain} we remove the dangerous $\hbar^{-\frac23}$ terms from the first two lines here whence
\begin{align}
 \partial_\zeta \tilde a_1(\zeta)  & = \Ai^{-2}(-\hbar^{-\frac23}\zeta) [1+\hbar a_2(\zeta,E)] \int_{-\hbar^{-\frac23}\zeta}^{\infty} \frac{\Ai}{\Bi}(\eta) (\partial_1 {\tilde a_2})(-\hbar^{\frac23}\eta,E)\,d\eta \nonumber\\
& \quad + \hbar \frac{\Bi}{\Ai}(-\hbar^{-\frac23}\zeta) \dot a_2(\zeta,E)  \int_{-\hbar^{-\frac23}\zeta}^{\infty}
  \Bi^{-2}(\eta) \tilde a_2(\hbar^{\frac23}\eta,E) \,
  d\eta  \label{eq:schl}
\end{align}
The contribution by the first integral here is treated as the integral in~\eqref{eq:pain} and is bounded
by
\begin{align*}
 &\les \la \hbar^{-\frac23}\zeta\ra^{\frac12} \Big[ \hbar^{-\frac13}\la \hbar^{-\frac23}\zeta\ra^{-1}  \chi_{[-1\le \zeta\le0]} + \hbar^{\frac13} \chi_{[\zeta_0\le\zeta\le-1]} \Big] \\
&\sim \hbar^{-\frac13}\la \hbar^{-\frac23}\zeta\ra^{-\frac12}  \chi_{[-1\le \zeta\le0]} +  |\zeta|^{\frac12} \chi_{[\zeta_0\le\zeta\le-1]}
\end{align*}
which is exactly as needed. Finally, the contribution of~\eqref{eq:schl} is bounded by
\begin{align*}
 &\les \hbar \Big[ \chi_{[-1\le \zeta\le0]} + |\zeta|^{\frac12} \chi_{[\zeta_0\le\zeta\le-1]} \Big] \\
&\sim \hbar \la\zeta\ra^{\frac12} \les  \hbar^{-\frac13}\la \hbar^{-\frac23}\zeta\ra^{-\frac12}  \chi_{[-1\le \zeta\le0]} +  |\zeta|^{\frac12} \chi_{[\zeta_0\le\zeta\le-1]}
\end{align*}
and we are done with the $k=0$ case of~\eqref{eq:aj_est} for $a_1$. However, since $E$ enters
into $a_1$ only through $a_2, \tilde a_2$ which do satisfy \eqref{eq:aj_est} for all $k\ge0$, we see that the previous estimates
carry over unchanged and provide the stated estimates for $\partial_E^k a_1(\zeta,E)$ and $\partial_E^k \partial_\zeta a_1(\zeta,E)$.
\end{proof}

\noindent We remark that the method employed in the previous proof does not extend easily to
derivatives $\partial_\zeta^\ell a_j$ with $\ell\ge2$ (that is, without losing excessive powers of $\hbar^{-1}$).
In principle, it is possible to treat $\ell=2$ by a similar method, but instead of a sharp cut-off  at $\zeta=0$ one needs to
use a smooth cut-off function in~\eqref{eq:a2}. However, the calculations  are quite involved and it is not clear
how to extend this approach systematically to higher~$\ell$ (the same comment applies to Proposition~\ref{prop:AiryII} below).
On the other hand, for the purposes of Theorem~\ref{thm:main}, as well as
for those of~\cite{SSS1} and~\cite{SSS2}, it suffices to treat the first derivative in~$\zeta$ (however, we do need many
derivatives relative to~$E$).
Next, we turn to $\zeta\ge0$ which requires an oscillatory basis.

\begin{prop}
  \label{prop:AiryII} Let $\hbar_0>0$ be small. In the range
$\zeta\ge 0$ a basis of solutions to~\eqref{eq:Airy} is given by
\begin{equation}\nonumber
\begin{aligned}
  \psi_1(\zeta,E;\hbar) &=
  (\Ai(\tau)+i\Bi(\tau)) [1+\hbar b_1(\zeta,E;\hbar)] \\
  \psi_2(\zeta,E;\hbar) &=
  (\Ai(\tau)-i\Bi(\tau)) [1+\hbar b_2(\zeta,E;\hbar)]
\end{aligned}
\end{equation}
 with
$\tau:=-\hbar^{-\frac23}\zeta$ and where
$b_1, b_2$  are smooth, complex-valued, and satisfy the bounds
for all $k\ge0$, and $j=1,2$
\begin{equation}\label{eq:bj_est}
\begin{aligned}
 | \partial_E^k\, b_j(\zeta,E;\hbar)| &\le C_{k}\, E^{-k} \la\zeta\ra^{-\frac32} \\
|\partial_\zeta\partial_E^k b_j(\zeta,E)| &\le C_{k}\, E^{-k} \hbar^{-\frac13} \la \hbar^{-\frac23} \zeta \ra^{-\frac12} \la \zeta\ra^{-2}
\end{aligned}
\end{equation}
uniformly in
the parameters $0<\hbar<\hbar_0$, $0<E<E_0$, $\zeta\ge0$.
\end{prop}
\begin{proof}
Let $\psi_{1,0}(\zeta;\hbar):=(\Ai+i\Bi)(\tau)$ and
$\psi_{2,0}(\zeta;\hbar):=(\Ai-i\Bi)(\tau)$.
We seek a basis of the form (dropping $\hbar$  as an independent variable from the notation)
\[
\psi_j(\zeta)=\psi_j(\zeta,E) = \psi_{j,0}(\zeta)(1+ \hbar
b_{j}(\zeta,E))
\] for $\zeta\ge0$.
This representation is meaningful since $\Ai$ and $\Bi$ have no common zeros (as
their Wronskian does not vanish). We obtain the
equation
\begin{equation}
  \label{eq:mult_ODEi} (\psi_{j,0}^2 \dot b_j)^{\dot{}} = - \frac{1}{\hbar } \tilde V \psi_{j,0}^2
(1+ \hbar b_{j})
\end{equation} for $j=1,2$, c.f.~\eqref{eq:mult_ODE}. A solution of~\eqref{eq:mult_ODEi} on
$\zeta\ge0$ is given by, with $b_j(\zeta)=b_j(\zeta,E)$,
\begin{align}
  b_j(\zeta) &:= \frac{-1}{\hbar } \int_\zeta^\infty
\psi_{j,0}^2(\eta) \int_{\zeta}^{\eta}
\psi_{j,0}^{-2}(\tilde\eta)\,d\tilde\eta\,
\tilde V(\eta,E)
(1+ \hbar b_{j}(\eta))\, d\eta
\label{eq:b2}
\end{align}
Recall the asymptotic behavior, see~\cite{Olver},
\begin{equation}
\Ai(-x)\pm i\Bi(-x) = \frac{1}{\pi^{\frac12} x^\frac14} e^{\mp i(\xi-\frac{\pi}{4})} (1+O(\xi^{-1}))
\label{eq:Aiasymp}
\end{equation}
as $x\to\infty$. Here $\xi=\frac23 x^{\frac32}$ and the $O(\cdot)$ term is complex-valued and exhibits symbol behavior:
\[
 \partial_\xi^k O(\xi^{-1}) = O(\xi^{-1-k}) \quad \forall\; k\ge0
\]
Therefore,  for any  $0>x_0>x_1$,
\[
\Big| (\Ai+i\Bi)^2(x_1) \int_{x_0}^{x_1} (\Ai+i\Bi)^{-2}(y)\, dy  \Big| \les \la
x_1\ra^{-\frac12}
\]
The leading term in~\eqref{eq:b2}, i.e.,
\[
b_{1,0}(\zeta,\hbar,E) :=
\hbar^{\frac13}\int_{\hbar^{-\frac23}\zeta}^\infty (\Ai+i\Bi)^2(-u)
\Big[\int_{\hbar^{-\frac23}\zeta}^u (\Ai+i\Bi)^{-2}(-v)\, dv\Big]\,
\tilde V(-\hbar^{\frac23}u,E) \, du
\]
therefore satisfies the bound, see  Lemma~\ref{lem:zeta2},
\begin{align*}
|b_{1,0}(\zeta)| &  \les \hbar^{\frac13}\int_{\hbar^{-\frac23}\zeta}^{\infty}
\la u\ra^{-\frac12}  |\tilde V(-\hbar^{\frac23}u,E)| \, du \\
&\les \hbar^{\frac13} \int_{\hbar^{-\frac23}\zeta}^{\infty}
\la u\ra^{-\frac12}  \la\hbar^{\frac23}u\ra^{-2} \, du
\lesssim   \la\zeta\ra^{-\frac32}
\end{align*}
uniformly in $\zeta\ge0$, $0<E<E_0$, and $0<\hbar<\hbar_0$. Due to
the linear nature of~\eqref{eq:b2},  a contraction argument now
yields the same bound for $b_1$; in fact, due to the derivative
bounds of Lemma~\ref{lem:zeta2} relative to~$E$, we obtain the more general estimate
\[
|\partial_E^k b_{j}(\zeta,E)| \le C_k\, E^{-k} \la\zeta\ra^{-\frac32} \quad \forall \;k\ge0
\]
uniformly in the parameters for both $j=1,2$. As for the first derivative in $\zeta$,
observe that
\begin{equation}\label{eq:bminprim}
\dot b_j(\zeta) = \frac{\hbar^{-1}}{
\psi_{j,0}^{2}(\zeta)} \int_\zeta^\infty
\psi_{j,0}^2(\eta)
 \tilde
V(\eta,E) (1+ \hbar b_{j}(\eta))\, d\eta
\end{equation}
In order to exploit the cancellation in this integral, one integrates by parts once.
To this end, write for $u\ge0$,
\[\begin{split}
 (\Ai+i\Bi)^2(-u) &= e^{\frac{4i}{3}\la u\ra^{\frac32}} \omega(u),\qquad |\omega(u)|\les \la u\ra^{-\frac12},\;|\omega'(u)|\les \la u\ra^{-\frac32}, \\
 \psi_{1,0}^2(\zeta;\hbar) &= e^{\frac{4i}{3}\la \hbar^{-\frac23} \zeta \ra^{\frac32}} \omega(\hbar^{-\frac23} \zeta)
\end{split}\]
Since
\[
 \psi_{1,0}^2(\zeta;\hbar)\,d\zeta = \frac{1}{2i} \hbar^{\frac23} \la \hbar^{-\frac23} \zeta \ra^{-\frac12}  \omega(\hbar^{-\frac23} \zeta)\,  d\Big[e^{\frac{4i}{3}\la \hbar^{-\frac23} \zeta \ra^{\frac32}}\Big]
\]
integration by parts yields
\begin{align}
 \dot b_1(\zeta) &= \frac{\hbar^{-\frac13}}{2i
\psi_{j,0}^{2}(\zeta)} \int_\zeta^\infty
\la \hbar^{-\frac23} \eta \ra^{-\frac12}  \omega(\hbar^{-\frac23} \eta)
 \tilde
V(\eta,E) (1+ \hbar b_{j}(\eta))\,  d\Big[e^{\frac{4i}{3}\la \hbar^{-\frac23} \eta \ra^{\frac32}}\Big]  \nonumber\\
&=  O\big(\hbar^{-\frac13} \la \hbar^{-\frac23} \zeta \ra^{-\frac12} \la \zeta\ra^{-2}  \big) - \label{eq:derivb1}\\
&\qquad\qquad
 - \frac{\hbar^{\frac23}}{2i
\psi_{j,0}^{2}(\zeta)} \int_\zeta^\infty e^{\frac{4i}{3}\la \hbar^{-\frac23} \eta \ra^{\frac32}}
O\big( \la \hbar^{-\frac23} \eta \ra^{-1}  \la\eta\ra^{-2}\big)
 \dot b_{1}(\eta)   \, d\eta  \label{eq:derivb1'}
\end{align}
The leading order here is given by \eqref{eq:derivb1}; indeed, if we
estimate the $\dot b_{1}(\eta)$ term in~\eqref{eq:derivb1'} by~\eqref{eq:derivb1}, then
$
 \eqref{eq:derivb1'} \less \hbar \la\zeta\ra^{-4},
$
which is much better than~\eqref{eq:derivb1}.
The conclusion is that
\[
 |\partial_\zeta\partial_E^k b_j(\zeta,E)|\les E^{-k} \hbar^{-\frac13} \la \hbar^{-\frac23} \zeta \ra^{-\frac12} \la \zeta\ra^{-2}
\]
as claimed.
\end{proof}

\section{The proof of Theorem~\ref{thm:main}}\label{sec:proof}

Let $f_{\pm}(x,E;\hbar)$ be the Jost solutions of $P(x,\hbar D)$ from~\eqref{eq:semiclass}. For ease
of notation, we shall first assume the symmetry $V(x)=V(-x)$ and later indicate how to treat the general case.
Also, as usual, we drop $\hbar$ from the arguments of functions.
Then $ f_-(x,E)=f_+(-x,E)$ so that the Wronskian of $f_+, f_-$ is
\[
 W(E)= -2f_+(0,E )f_+'(0,E )
\]
Next, from \eqref{eq:Aiasymp},  and with $\zeta=\zeta(x,E)$ as in~\eqref{eq:zeta} and $T_+(E)$ as in \eqref{eq:STdef},
\[
 f_+(x,E ) = \sqrt{\pi}\,E^{\frac14}\hbar^{-\frac16}e^{i(\frac{T_+(E)}{\hbar}+\frac{\pi}{4})} q^{-\frac14}(\zeta) \psi_2(\zeta,E )
\]
This is obtained by matching the asymptotic behavior of $f_+$ with that of $\psi_2(\zeta)$ as $x\to\infty$
and we used the relation $w=q^{\frac14}f$ from Lemma~\ref{lem:langer}.
We now connect $\psi_2$ to the basis $\phi_j(\zeta,E )$ of Proposition~\ref{prop:AiryI}:
\[
 \psi_2(\zeta,E ) = c_1(E ) \phi_1(\zeta,E ) + c_2(E ) \phi_2(\zeta,E )
\]
where
\[
 c_1(E ) = \frac{W(\psi_2(\cdot,E ),\phi_2(\cdot,E ))}{W(\phi_1(\cdot,E ), \phi_2(\cdot,E ))}, \quad c_2(E ) = - \frac{W(\psi_2(\cdot,E ),\phi_1(\cdot,E ))}{W(\phi_1(\cdot,E ), \phi_2(\cdot,E ))}
\]
By Proposition~\ref{prop:AiryI},
\[
 W(\phi_1(\cdot,E ), \phi_2(\cdot,E )) = -\hbar^{-\frac23} W(\Ai,\Bi) + O(\hbar^{-\frac23}) =-\pi^{-1}\hbar^{-\frac23}(1+O(\hbar))
\]
where we evaluated the Wronskian on the left-hand side at $\zeta=0$. Next, by Propositions~\ref{prop:AiryI} and~\ref{prop:AiryII},
\begin{equation}
 \label{eq:Wphipsi}
\begin{aligned}
  W(\psi_2(\cdot,E ),\phi_2(\cdot,E )) &= -\hbar^{-\frac23}[(\Ai(0)-i\Bi(0))\Bi'(0)\\
&\qquad -(\Ai'(0)-i\Bi'(0))\Bi(0)+O(\hbar)] \\
&= -\hbar^{-\frac23}[W(\Ai,\Bi) +O(\hbar)] \\
W(\psi_1(\cdot,E ),\phi_1(\cdot,E )) &= -\hbar^{-\frac23}[(\Ai(0)-i\Bi(0))\Ai'(0)\\
&\qquad -(\Ai'(0)-i\Bi'(0))\Ai(0)+O(\hbar)] \\
&= -\hbar^{-\frac23}[iW(\Ai,\Bi) +O(\hbar)]
\end{aligned}
\end{equation}
so that
\begin{equation}\label{eq:c1c2}
 c_1(E )=1+O(\hbar),\quad c_2(E )=-i+O(\hbar)
\end{equation}
where the $O(\cdot)$ terms satisfy $|\partial_E^k O(\hbar)|\le C_k\, E^{-k}\ $.
For the remainder of the proof, we set $\zeta_0:=\zeta(0,E)$.  Then
\begin{align*}
 f_+(0,E ) &= \sqrt{\pi}e^{i(\frac{T_+(E)}{\hbar}+\frac{\pi}{4})}\,E^{\frac14}\hbar^{-\frac16} q^{-\frac14}(\zeta_0)\psi_2(\zeta_0,E ) \\
&= \sqrt{\pi}e^{i(\frac{T_+(E)}{\hbar}+\frac{\pi}{4})}\,E^{\frac14}\hbar^{-\frac16} q^{-\frac14}(\zeta_0) [ c_1(E )\phi_1(\zeta_0,E ) + c_2(E ) \phi_2(\zeta_0,E ) ]\\
 f_+'(0,E ) &= \sqrt{\pi}e^{i(\frac{T_+(E)}{\hbar}+\frac{\pi}{4})}\,E^{\frac14}\hbar^{-\frac16}\zeta'(0)q^{-\frac14}(\zeta_0)\big[\psi_2'(\zeta_0,E )-\frac14 \frac{\dot q}{q}(\zeta_0)\psi_2(\zeta_0,E )  \big] \\
&=\sqrt{\pi}e^{i(\frac{T_+(E)}{\hbar}+\frac{\pi}{4})}\,E^{\frac14}\hbar^{-\frac16}\zeta'(0)q^{-\frac14}(\zeta_0)\big[c_1(E )\phi_1'(\zeta_0,E ) + c_2(E ) \phi_2'(\zeta_0,E ) \\
&\quad -\frac14 \frac{\dot q}{q}(\zeta_0)(c_1(E )\phi_1(\zeta_0,E ) + c_2(E ) \phi_2(\zeta_0,E ) )  \big]
\end{align*}
Recall from Lemma~\ref{lem:langer} that $\zeta'=q^{\frac12}$.
From $V'(0)=0$ we obtain
\[
 \dot q(\zeta_0)=\frac{Q_0(0)}{\zeta_0^2} = -\frac{q(\zeta_0)}{\zeta_0}
\]
and thus
\begin{align*}
& f_+(0,E )f_+'(0,E ) = i\pi E^{\frac12} e^{2i\frac{T_+(E)}{\hbar}}\hbar^{-\frac13}[ c_1(E )\phi_1(\zeta_0,E ) + c_2(E ) \phi_2(\zeta_0,E ) ]\times\\
&\qquad\qquad\qquad\qquad\qquad\qquad \times [c_1(E )\phi_1'(\zeta_0,E ) + c_2(E ) \phi_2'(\zeta_0,E )] \\
& \quad\qquad\qquad\qquad\qquad + \frac{i}{4}\pi E^{\frac12}e^{2i\frac{T_+(E)}{\hbar}} \hbar^{-\frac13}\zeta_0^{-1}[ c_1(E )\phi_1(\zeta_0,E ) + c_2(E ) \phi_2(\zeta_0,E ) ]^2
\end{align*}
From Proposition~\ref{prop:AiryI},
\begin{align*}
\phi_1(\zeta_0,E ) &= \Ai(-\hbar^{-\frac23}\zeta_0)(1+O(\hbar))\\
 \phi_2(\zeta_0,E ) &=  \Bi(-\hbar^{-\frac23}\zeta_0)(1+O(\hbar))\\
 \phi_1'(\zeta_0,E ) &= -\hbar^{-\frac23}\Ai'(-\hbar^{-\frac23}\zeta_0)(1+O(\hbar))+O(\hbar)|\zeta_0|^{\frac12}\Ai(-\hbar^{-\frac23}\zeta_0)\\
 \phi_2'(\zeta_0,E ) &= -\hbar^{-\frac23}\Bi'(-\hbar^{-\frac23}\zeta_0)(1+O(\hbar))+O(\hbar)|\zeta_0|^{\frac12}\Bi(-\hbar^{-\frac23}\zeta_0)
\end{align*}
which implies via the standard asymptotics of the Airy functions that
\begin{align*}
\phi_1(\zeta_0,E ) &= (4\pi)^{-\frac12}(\hbar^{-\frac23}|\zeta_0|)^{-\frac14}e^{-\frac23 \hbar^{-1}|\zeta_0|^{\frac32}}(1+O(\hbar))\\
 \phi_2(\zeta_0,E ) &= \pi^{-\frac12}(\hbar^{-\frac23}|\zeta_0|)^{-\frac14}e^{\frac23 \hbar^{-1}|\zeta_0|^{\frac32}}(1+O(\hbar))\\
 \phi_1'(\zeta_0,E ) &= \hbar^{-\frac23}(4\pi)^{-\frac12}(\hbar^{-\frac23}|\zeta_0|)^{\frac14}e^{-\frac23 \hbar^{-1}|\zeta_0|^{\frac32}}(1+O(\hbar))\\
 \phi_2'(\zeta_0,E ) &= -\hbar^{-\frac23}\pi^{-\frac12}(\hbar^{-\frac23}|\zeta_0|)^{\frac14}e^{\frac23 \hbar^{-1}|\zeta_0|^{\frac32}} (1+O(\hbar))
\end{align*}
Hence, using that $e^{-\hbar^{-1}|\zeta_0|^{\frac32}} = O(\hbar)$ where $\partial_E^k O(\hbar) = O(E^{-k}\hbar)$,
one concludes that
\begin{align*}
& \hbar^{-\frac13}[ c_1(E )\phi_1(\zeta_0,E ) + c_2(E ) \phi_2(\zeta_0,E ) ]\times\\
&\qquad\qquad\qquad\qquad\qquad\qquad\times [c_1(E )\phi_1'(\zeta_0,E ) + c_2(E ) \phi_2'(\zeta_0,E )] \\
& =  \pi^{-1} \hbar^{-1} e^{\frac43\hbar^{-1}|\zeta_0|^{\frac32}}(1+O(\hbar))
\end{align*}
as well as
\[
 \hbar^{-\frac13}\zeta_0^{-1}[ c_1(E )\phi_1(\zeta_0,E ) + c_2(E ) \phi_2(\zeta_0,E ) ]^2
 = -\pi^{-1}|\zeta_0|^{-\frac32}e^{\frac43\hbar^{-1}|\zeta_0|^{\frac32}}(1+O(\hbar))
\]
Since $T(E)=2T_+(E)$ we finally arrive at
\begin{align}
\nonumber W(E)=-2f_+(0,E )f_+'(0,E ) &= -2i e^{2i\frac{T_+(E)}{\hbar}}E^{\frac12}\hbar^{-1} e^{\frac43\hbar^{-1}|\zeta_0|^{\frac32}}(1+O(\hbar))  \\
&=  -\frac{2i\sqrt{E}}{\hbar} e^{\hbar^{-1}(S(E)+iT(E))} (1+O(\hbar)) \label{eq:Wasymp}
\end{align}
We used here that
\[
 \frac43|\zeta_0|^{\frac32} = 2\int_0^{x_1}\sqrt{V_0(\eta)-E}\,d\eta = S(E)
\]
All the $O(\hbar)$ appearing above behave as required under differentiation with respect to~$E$;
indeed, this is both due to the bounds of Propositions~\ref{prop:AiryI} and~\ref{prop:AiryII} as well as the aforementioned fact that
\[
 e^{-\frac23\hbar^{-1}|\zeta_0|^{\frac32}} = O(\hbar |\zeta_0|^{-\frac32}) = O(\hbar)
\]
has the required behavior
since $|\zeta_0|^{-\frac32} = O(|\log E|^{-\frac32})$ as $E\to0+$.
In view of~\eqref{eq:Sigma11W}, \eqref{eq:Wasymp} implies the sought after asymptotic relation for~$\Sigma_{11}$
in Theorem~\ref{thm:main}, see~\eqref{eq:sentries}.

\noindent In order to find $\Sigma_{12}$, and $\Sigma_{21}$ (i.e., the reflection coefficients), we need to also asymptotically
evaluate the following Wronskians:
\[
 W(f_{+}(\cdot, E ),\overline{f_{-}(\cdot, E )})=  W(\overline{f_{+}(\cdot, E )},f_{-}(\cdot, E ))=-2\,\text{Re}\,[f_+(0,E )\overline{f_+'(0,E )}].
\]
Using the same notations as in the computation of $W(E)$, we obtain
\begin{align*}
-2\textrm{Re}\,[ f_+(0,E )\overline{f_+'(0,E )}] &= -2\textrm{Re}\,\big\{\pi E^{\frac12} \hbar^{-\frac13}[ c_1(E )\phi_1(\zeta_0,E ) + c_2(E ) \phi_2(\zeta_0,E ) ]\times\\
&\qquad \times [\overline{c_1(E )}\phi_1'(\zeta_0,E ) + \overline{c_2(E )} \phi_2'(\zeta_0,E )]\big\} \\
&\qquad  - \frac{\pi}{2} E^{\frac12}\hbar^{-\frac13}\zeta_0^{-1}\vert c_1(E )\phi_1(\zeta_0,E ) + c_2(E ) \phi_2(\zeta_0,E ) \vert^2
\end{align*}
Finally, evaluating this expression as above, we obtain
\[
W(f_{+}(\cdot, E ),\overline{f_{-}(\cdot, E )})= -2\textrm{Re}\, [f_+(0,E )\overline{f_+'(0,E )}]  = \frac{2 \sqrt{E}}{\hbar} e^{\frac{S(E)}{\hbar}}(1+O(\hbar)).
\]
Forming the ratio between this formula and the one for $W(E)$ yields the desired expression for $\Sigma_{12}=\Sigma_{21}$, see~\eqref{eq:sentries}. Indeed,
\[
 r_{-}(E)=-\frac{W(\overline{f_{+}(\cdot, E )},f_{-}(\cdot, E ))}{ W(f_{+}(\cdot, E ),{f_{-}(\cdot, E )})} = -ie^{-i\hbar^{-1}T(E)}(1+O(\hbar))
\]
where $O(\hbar)$ behaves like a symbol with respect to~$E$, as usual. 
This concludes the proof of Theorem~\ref{thm:main} in the symmetric case. If $V(x)\ne V(-x)$, then only minor changes are
needed. Indeed, on $x\le 0$ we can still use the {\em same bases} $\phi_j, \psi_j$  from Section~\ref{sec:Airy} but
with $\zeta=\zeta(-x,E)$. This is due to the fact that the difference between the left-hand and right-hand branches of $V$ does not
affect the estimates from Section~\ref{sec:Airy} (since we are assuming inverse square decay at both ends and the constants $\mu_{\pm}$ have no effect on the leading order behavior).  Let
\[
  \tilde\zeta_0(E)^{\frac32} := \frac32 \int_{x_2(E)}^0 \sqrt{V_0(\eta)-E}\, d\eta
\]
Thus, in addition to the expressions for $f_+(0,E)$ and $f_+'(0,E)$
from above we now also have
\begin{align*}
 f_-(0,E )
&= \sqrt{\pi}e^{i(\frac{T_-(E)}{\hbar}+\frac{\pi}{4})}\,E^{\frac14}\hbar^{-\frac16} q^{-\frac14}(\tilde\zeta_0) [ c_1(E )\phi_1(\tilde\zeta_0,E ) + c_2(E ) \phi_2(\tilde\zeta_0,E ) ]\\
 f_-'(0,E ) &= - \sqrt{\pi}e^{i(\frac{T_-(E)}{\hbar}+\frac{\pi}{4})}\,E^{\frac14}\hbar^{-\frac16}\zeta'(0)q^{-\frac14}(\tilde\zeta_0)\big[c_1(E )\phi_1'(\tilde\zeta_0,E ) + c_2(E ) \phi_2'(\tilde\zeta_0,E ) \\
&\quad -\frac14 \frac{\dot q}{q}(\tilde\zeta_0)(c_1(E )\phi_1(\tilde\zeta_0,E ) + c_2(E ) \phi_2(\tilde\zeta_0,E ) )  \big]
\end{align*}
Inserting these expressions into
\[
 W(E) =  f_+(0,E ) f_-'(0,E ) - f_+'(0,E )f_-(0,E ),
\]
and using that
\[
 \frac23\big[\zeta_0^{\frac32} + \tilde\zeta_0^{\frac32}\big] =  \int_{x_2(E)}^{x_1(E)} \sqrt{V_0(\eta)-E}\, d\eta  =S(E)
\]
as well as $T(E)=T_+(E)+T_-(E)$, one again arrives at \eqref{eq:Wasymp}. The same comments apply to the off-diagonal
terms of the scattering matrix and we are done. As for the very last claim of the theorem concerning $V_0=V+\hbar^2 V_1$,
simply note that the main calculations entering into the above proof only make use of the leading order part of~$V_1$,
i.e., $\frac14\la x\ra^{-2}$ whereas the cubic piece gets absorbed into the error term. 

\section{From small to large energies}\label{sec:largeE}

In this section, we present an extension of
Theorem~\ref{thm:main} to the case of large energies. More specifically,  suppose $V$ is as in
Theorem~\ref{thm:main} but with the following additional properties:
\begin{itemize}
 \item $0<V(x)\le 1$ for all $x\in\R$, $V(0)=1$, $V'(0)=0$, $V''(0)=-1$
\item $V'(x)<0$ for all $x>0$, $V'(x)>0$ for all $x<0$
\end{itemize}
Note that this is precisely the kind of barrier potential considered by Ramond~\cite{Ramond} (but without any
analyticity assumptions).  For the purposes of this section we refer to it as a {\em simple barrier potential}.  
Even though Theorem~\ref{thm:main} by design only considers {\em small} energies $0<E<E_0$, it is natural to ask to what
extent it remains correct as $E_0\to1$. As already remarked before, for energies $E>\eps>0$ there is no difference between
$V$ and $V_0$ as far as Theorem~\ref{thm:main} is concerned. Indeed, switching from $V$ to $V_0$ only affects the error term. 
Moreover, for the kind of $V$ we are considering here, the theorem remains valid in any range $0<E<1-\eps$ with $\eps$ fixed.
This is due to the fact that in this range there is a unique pair of turning points $x_2(E), x_1(E)$ as before. The action $S(E;\hbar)$
lies between two positive constants (depending on~$\eps$) and the previous proof goes through without changes. Somewhat more
interesting and very relevant for  later applications, cf.~\cite{SSS1}, \cite{SSS2}, is the case where $\eps=\hbar^{\alpha}$. The
question is then how large  $\alpha\ge0$ can be allowed to be. First note that we can no longer expect the error term in~\eqref{eq:sentries} to be of the form $O(\hbar)$ in that case. Rather, it will  need to 
be $O(\hbar^\delta)$ for some $\delta=\delta(\alpha)>0$ and this condition will determine how large we can take $\alpha$. It
turns out that the range~$0\le \alpha<1$ is admissible here. In the following corollary, we use the notations introduced in Theorem~\ref{thm:main}. 

\begin{cor}\label{cor:main} 
Let $V$ be a simple barrier potential. 
For every $0<\alpha<1$ there exists  and $\hbar_0=\hbar_0(\alpha)$ small such that 
 for all $0<\hbar<\hbar_0$  and $0<E\le 1-\hbar^\alpha$
\begin{equation}\label{eq:sentries2}
\begin{aligned}
 \Sigma_{11}(E;\hbar) &= e^{-\frac{1}{\hbar}(S(E;\hbar)+iT(E;\hbar))} \big(1+ \hbar(1-E)^{-1}\,\sigma_{11}(E;\hbar)\big) \\
 \Sigma_{12} (E;\hbar)  &= -i e^{-\frac{2i}{\hbar} T_+(E;\hbar)} \big(1+\hbar (1-E)^{-1}\, \sigma_{12}(E;\hbar)\big)
\end{aligned}
\end{equation}
and the correction terms satisfy the bounds
\begin{equation}
\label{eq:errors2}
 |\partial_E^k\, \sigma_{11}(E;\hbar)|+|\partial_E^k\, \sigma_{12}(E;\hbar)| \le C_k\, \max(E^{-k},(1-E)^{-k})\quad\forall\; k\ge0,
\end{equation}
with a constant $C_k$ that only depends on $k$ and $V$.
\end{cor}
\begin{proof} We will only sketch the proof as there is no point in repeating all the details of the proof of Theorem~\ref{thm:main}. 
In fact, inspection of the previous section shows that the main issue is to prove that Propositions~\ref{prop:AiryI} and~\ref{prop:AiryII}
remain valid albeit with errors of the form~$\hbar^{1-\alpha}$ rather than~$\hbar$ (we need to pay particular attention to the derivative $\partial_\zeta$) when $E=1-\hbar^\alpha$. 
We will freely use the notation from Section~\ref{sec:Langer} and~\ref{sec:Airy}. 
 By the preceding comments, it will suffice to consider the range $1-\eps<E\le 1-\hbar^\alpha$. In fact, it will be enough to set $E=1-\hbar^\alpha$ so that $x_1(E)\sim \hbar^{\frac{\alpha}{2}}$. The range $0<x<x_1(E)$ then corresponds to the region 
$-\hbar^{\frac{2\alpha}{3}}\les \zeta\le 0$. A simple calculation shows that $q\sim \hbar^{\frac{\alpha}{3}}$ in that range, as
well as $|\tilde V|\les \hbar^{-\frac{4\alpha}{3}}$ with the usual behavior under differentiation in~$E$.  In fact, for all 
$0\le x\le x_1(E)$ we have
\[
 V(x)-E= -\int_{x}^{x_1} V'(y)\, dy \sim x_1^2 - x^2 \sim (x_1-x)x_1
\]
and thus
\[
 \zeta \sim -x_1^{\frac13} (x_1-x), \qquad q=\frac{V-E}{-\zeta} \sim \frac{ x_1(x_1-x)}{x_1^{\frac13} (x_1-x)} = x_1^{\frac23}\sim \hbar^{\frac{\alpha}{3}}
\]
as claimed. Next, recall \eqref{eq:tilV}, viz. 
\begin{equation}
\tilde V=  \frac{1}{4}q^{-1} \langle x\rangle ^{-2}+\frac{3}{16} q^{-2}
\dot q^2 - \frac14 q^{-1} \ddot q
\end{equation}
Since $\dot q = q^{-\frac12} q'$ where $q'=\frac{dq}{dx}\sim x_1^{-\frac13}\sim q^{-\frac12}$, the second term here is of size
\[
 q^{-2}
\dot q^2 \les q^{-3} (q')^2 \les q^{-4}\sim \hbar^{-\frac{4\alpha}{3}}
\]
The other two terms are smaller whence $|\tilde V|\les \hbar^{-\frac{4\alpha}{3}}$ as claimed. 
Turning to Proposition~\ref{prop:AiryI}, we seek a basis of the form
\begin{equation}\nonumber
\begin{aligned}
  \phi_1(\zeta,E,\hbar) &=
  \Ai(\tau) [1+\hbar^\delta a_1(\zeta,E,\hbar)] \\
  \phi_2(\zeta,E,\hbar) &=
  \Bi(\tau) [1+\hbar^\delta a_2(\zeta,E,\hbar)]
\end{aligned}
\end{equation}
with $\delta:=1-\alpha$.  Proceeding as in the proof of Theorem~\ref{thm:main},
we arrive at the following analogue of~\eqref{eq:a20_bd}  
\[
 |a_{2,0}(\zeta)| \les \hbar^{\frac43-\delta}\int_{0}^{-\hbar^{-\frac23}\zeta}
\la u\ra^{-\frac12}  |\tilde V(-\hbar^{\frac23}u,E)| \, du 
\]
which yields
\[
 |a_{2,0}(\zeta)| \les \hbar^{\frac43-\delta} \la \hbar^{-\frac23}\zeta\ra^{\frac12} \hbar^{-\frac{4\alpha}{3}}
\les \hbar^{1-\alpha-\delta} 
\]
This shows that with our choice of $\delta$, we have
\[
 \sup_{\zeta(0,E)\le\zeta\le 0}|a_{2,0}(\zeta)|\les 1
\]
For the derivatives, the analogue of \eqref{eq:aminprim}, viz., 
\begin{equation}\nonumber
\dot a_2(\zeta) = \frac{\hbar^{-\delta}}{
\phi_{2,0}^{2}(\zeta,\hbar)} \int_\zeta^0
\phi_{2,0}^2(\eta,\hbar)
 \tilde
V(\eta,E) (1+ \hbar a_{2}(\eta))\, d\eta
\end{equation}
yields
\begin{align*}
|\dot a_2(\zeta,E,\hbar)| 
&\les \hbar^{\frac23-\delta} \Bi^{-2}(-\hbar^{-\frac23}\zeta) \int_{0}^{-\hbar^{-\frac23}\zeta} \Bi^2(u)\hbar^{-\frac{4\alpha}{3}}\, du \\
&\les \hbar^{\frac23-\delta-\frac{4\alpha}{3}} \la \hbar^{-\frac23}\zeta\ra^{-\frac12} \les \hbar^{-\frac23}
\end{align*}
where we again used that $\alpha<1$ in the final step. 
An analogous estimate holds for~$\phi_1$. 
We claim that these bounds are sufficient {\em provided}
the same type of estimates hold for the analogue of Proposition~\ref{prop:AiryII} at $\zeta=0$. Indeed, 
inspection of~\eqref{eq:Wphipsi} shows that in that case
\begin{align*}
  W(\psi_2(\cdot,E ),\phi_2(\cdot,E )) &= -\hbar^{-\frac23}[(\Ai(0)-i\Bi(0))\Bi'(0)\\
&\qquad -(\Ai'(0)-i\Bi'(0))\Bi(0)+O(\hbar^\delta)] \\
&= -\hbar^{-\frac23}[W(\Ai,\Bi) +O(\hbar^\delta)] \\
W(\psi_1(\cdot,E ),\phi_1(\cdot,E )) &= -\hbar^{-\frac23}[(\Ai(0)-i\Bi(0))\Ai'(0)\\
&\qquad -(\Ai'(0)-i\Bi'(0))\Ai(0)+O(\hbar^\delta)] \\
&= -\hbar^{-\frac23}[iW(\Ai,\Bi) +O(\hbar^\delta)]
\end{align*}
Note that there is an exact balance here between the $\hbar^{-\frac23}$ coming from the derivatives of the 
main contributions and the losses stemming from $\dot a_j, \dot b_j$. Hence, 
\[
 c_1(E )=1+O(\hbar^\delta),\quad c_2(E )=-i+O(\hbar^\delta)
\]
as desired. Since $\hbar^{-1}|\zeta_0|^{\frac32}\sim \hbar^{\alpha-1}$ and thus also
\[
 e^{-\hbar^{-1}|\zeta_0|^{\frac32}} = O(\hbar^{1-\alpha}) = O(\hbar^{\delta})
\]
the reader will easily check that the remainder of the proof in Section~\ref{sec:proof} goes through.

It therefore remains to deal with the oscillatory regime. In analogy with Proposition~\ref{prop:AiryII} we seek a basis 
\begin{equation}\nonumber
\begin{aligned}
  \psi_1(\zeta,E;\hbar) &=
  (\Ai(\tau)+i\Bi(\tau)) [1+\hbar^\delta b_1(\zeta,E;\hbar)] \\
  \psi_2(\zeta,E;\hbar) &=
  (\Ai(\tau)-i\Bi(\tau)) [1+\hbar^\delta b_2(\zeta,E;\hbar)]
\end{aligned}
\end{equation}
For this we need to understand $\tilde V$ on $\zeta\ge0$. First one checks that for
all $x\ge x_1(E)$, 
\[
 \zeta \sim \left\{ \begin{array}{ll} x_1^{\frac13}(x-x_1) & x_1\le x\le 2x_1 \\
                     x^{\frac43} & 2 x_1 \le x\ll 1 \\
                     x^{\frac23}  & x\gtrsim  1 
                    \end{array}\right.
\]
and thus 
\[
 q \sim \left\{ \begin{array}{ll} x^{\frac23} & x_1\le x\ll 1 \\
                                  x^{-\frac23} & x\gtrsim 1
                \end{array}\right.
\]
Hence, \eqref{eq:tilV} implies that
\[
 |\tilde V|\les x^{-\frac83}\chi_{[x_1\le x\le 1]} + \zeta^{-2}\chi_{[x\ge1]}
\]
Going through the proof of Proposition~\ref{prop:AiryII} shows that 
\[
 |b_j(0)|\les \hbar^{1-\alpha-\delta}\les 1, \quad |\dot b_j(0)|\les \hbar^{\frac23-\delta-\frac{4\alpha}{3}}\les \hbar^{-\frac23}
\]
as desired. The derivatives relative to $E$ are left to the reader. 
\end{proof}

Thus, the semi-classical approximation obtained in Theorem~\ref{thm:main} breaks down precisely at $E=1-\hbar$. 
As is well-known, the Airy equation is no longer the correct approximating equation  for energies close to the unique maximum $V(0)=1$
of a simple barrier potential.  In fact, there exists an analytic change of variables
which reduces the Schr\"odinger equation with such energies  to the Weber equation locally around the origin. 
Alternatively, Ramond~\cite{Ramond} invokes
micro-local methods and the Helffer-Sj\"ostrand normal form in that case.

\appendix

\section{A normal form reduction to Bessel's equation}

In this section we sketch an alternative route for the asymptotic analysis of Section~\ref{sec:Airy}.
It is based on Lemma~\ref{lem:eta} and reduces equation~\eqref{nonzeroenergy.eq} to
a Bessel equation rather than an Airy equation. However, we emphasize that these approaches
are in fact quite related as the Airy functions are used to describe Bessel functions $J_n$ and $Y_n$
in the large $n=\hbar^{-1}$ asymptotics very much in the spirit of Section~\ref{sec:Airy}, see~\cite{Olver}. A possible advantage of working with the Bessel
representation lies with the fact that they apply to all $x\in[\eps E^{-\frac12},\infty)$ which is
a region containing the turning point $x_1(E)$.
On the other hand, since they cannot be used on the region $[0,\eps E^{-\frac12}]$, one is again
faced with a connection problem as in Section~\ref{sec:Airy}. Moreover, we have found that using
distinct changes of variables in these two regions leads to a number of complications as compared to
the global action-based coordinates introduced in Lemma~\ref{lem:langer}.
For this, as well as other reasons, we
ultimately found it technically advantageous to
work with the Airy approximation directly, but we wish to sketch the Bessel method  since it seems
to be of independent interest.
 In this section, we shall use the notations of Lemma~\ref{lem:eta} and always work
on $y\ge1$ which transforms into $\xi\ge \xi_1(E)$, see~\eqref{eq:x1E}.
First, a preliminary technical lemma.

\begin{lem}
  \label{lem:muOmega} The function $\mu(\eta,E):=(\partial_\xi
y(\eta,E))^2(\partial_{yy} \xi)(y(\eta,E),E)$ satisfies
\begin{equation}
  \label{eq:mu_dec} |\partial_E^k\partial_\eta^j \mu(\eta,E)|\le C_{kj}\, E^{-k}\eta^{-j-3}
\end{equation}
and the positive smooth function
\begin{equation}\label{eq:Omega}
\Omega(\eta,E):= \exp\Big(-\int_\eta^\infty \mu(t,E)\, dt\Big)
\end{equation}
 satisfies $\frac{\Omega'}{\Omega}=\mu$ and
$\Omega=1+O(\eta^{-2})$ (we write $'=\partial_\eta$) with a symbol-type
$O(\eta^{-2})$.
\end{lem}
\begin{proof}
The $\eta^{-3}$ decay in \eqref{eq:mu_dec} is due to
\eqref{eq:yforlargex}. Otherwise, the lemma is an immediate
consequence of Lemma~\ref{lem:eta}.
\end{proof}

Now for the transformation of the equation with $V,V_0$ as in Theorem~\ref{thm:main}.
To motivate our way of obtaining the Bessel equation as an approximating equation, consider
the model operator
\[
 P(x,\hbar D):= -\hbar^2 \partial_x^2 + \la x\ra^{-2}
\]
It is tempting to introduce the Bessel operator
\[
 P_0(x,\hbar D):= -\hbar^2 \partial_x^2 + x^{-2}
\]
which should be a good approximation for large $x$. The problem here is that even though the error
decays like $x^{-4}$ it is not small compared to $\hbar$ unless $x>\hbar^{-\frac14}$. Since we need to
be able to send $\hbar$ and $E$ to zero {\em independently}, such an approximation is useless for the
case were $E$ is small but fixed and $\hbar\to0$. To idea behind our reduction to the Bessel equation
is essentially to let $\la x\ra$ be a new independent variable. The reader
will easily see that this is precisely what Lemma~\ref{lem:eta} does  (in addition, we scale out~$E$ and fix the
turning point to lie at~$1$). 

\begin{lem}
  \label{lem:transf} For any $0<E<E_0$ the following holds: $f(x)$
  is a smooth solution of
\[
-\hbar^2 f''(x) + V(x) f(x) = Ef(x)
\text{\ \ on\ \ }x> E^{-\frac12}
\]
iff $\phi(\xi) =\phi(\xi,E):= \Omega(\xi,E)^{\frac12} f(E^{-\frac12}
y(\xi,E))$, with $\Omega$ as in \eqref{eq:Omega},  is a smooth
solution of
\begin{equation}\label{eq:master}
 -\hbar^2 \phi''(\xi) + \big[\xi^{-2}(1-\hbar^2/4)-1  \big] \phi(\xi)  = \hbar^2 W_0(\xi,E)
 \phi(\xi) \text{\ \ on\ \ }\xi>\xi_1(E)
\end{equation}
with a potential $W_0$ satisfying
\begin{equation}\label{eq:W0}
|\partial_E^k \partial_\xi^\ell W_0(\xi,E)|\le C_{k,\ell} E^{-k}
\xi^{-3-\ell}
\end{equation}
for all $k,\ell\ge0$.
\end{lem}
\begin{proof}
Under the change of variables $g(y)=f(E^{-\frac12}y)$ the following
equations are equivalent, with $V_1(x)=-\frac14\la x\ra^{-2}$:
\begin{align*}
-\hbar^2 f''(x) + (V_0(x)+\hbar^2 V_1(x)) f(x) &= Ef(x) \\
 -\hbar^2 g''(y) +  (E^{-1}V_0(E^{-\frac12}y)-1) g(y) &= -\hbar^2 E^{-1}
 V_1(E^{-\frac12} y) g(y)
\end{align*}
Now let $\xi=\xi(y,E)$ be as in Lemma~\ref{lem:eta} and set
$\psi(\xi)=g(y(\xi,E))$, or equivalently, $\psi(\xi(y,E))=g(y)$. Then,
with $\mu$ as in Lemma~\ref{lem:muOmega},
\begin{equation}
  \label{eq:premaster}
  \begin{aligned}
& -\hbar^2 [\psi''(\xi) + (\partial_\xi y(\xi,E))^2(\partial_{yy}
\xi)(y(\xi,E),E)\psi'(\xi)]+ (\xi^{-2}-1) \psi(\xi) \\
& =-\hbar^2 [\psi''(\xi) + \mu(\xi,E)\psi'(\xi)]+ (\xi^{-2}-1) \psi(\xi) \\
& = - \hbar^2 (\partial_\xi y(\xi,E))^2 E^{-1}
 V_1(E^{-\frac12} y(\xi,E)) \psi(\xi)
\end{aligned}
\end{equation}
Let $\Omega$ be as in Lemma~\ref{lem:muOmega}. In view
of~\eqref{eq:premaster}, $\phi:= \Omega^{\frac12} \psi$ satisfies
the equation
\begin{equation}
  \label{eq:premaster2}   \begin{aligned}
& -\hbar^2 \phi''(\xi) + (\xi^{-2}-1) \phi(\xi) & = \hbar^2 W(\xi,E) \phi(\xi)
\end{aligned}
\end{equation}
with
\[
W(\xi,E) := -(\partial_\xi y(\xi,E))^2 E^{-1}
 V_1(E^{-\frac12} y(\xi,E)) -\frac12 \frac{\Omega''(\xi,E)}{\Omega(\xi,E)} + \frac14\Big(
\frac{\Omega'(\xi,E)}{\Omega(\xi,E)}\Big)^2
\]
The second part here involving $\Omega$ decays like $\xi^{-4}$,
whereas the first only decays like $\xi^{-2}$. We need to extract this
leading order decay: the asymptotic expansion
\[
V_1(\xi) = -\frac{1}{4\xi^2} + O(\xi^{-3})\text{\ \ as\ \
}\xi\to\infty
\]
and Lemma~\ref{lem:eta} imply that
\[
(\partial_\xi y(\xi,E))^2 E^{-1}
 V_1(E^{-\frac12} y(\xi,E)) = -\frac{1}{4\xi^2} + \xi^{-3} V_r(\xi,E)
\]
where
\[
|\partial_E^k \partial_\xi^\ell V_r(\xi,E)|\le C_{k,\ell} E^{-k}
\xi^{-\ell}
\]
In view of~\eqref{eq:premaster2}, this yields equation
\eqref{eq:master} and we are done.
\end{proof}

A fundamental system $\{\phi_{j,n}^{(0)}\}_{j=1}^2$ of the
homogeneous form of \eqref{eq:master}, i.e.,
\begin{equation}\label{eq:master_hom} -\hbar^2 \phi''(\xi) +
\big[\xi^{-2}(1-\hbar^2/4)-1  \big] \phi(\xi)  =0
\end{equation}
is given in terms of Hankel functions:
\[
\phi_{j,n}^{(0)}(\xi) = \xi^{\frac12} H_n^{(j)}(n\xi), \qquad j=1,2,\;
n:=\hbar^{-1}
\]
or, equivalently, by the Bessel functions
\[
\wt\phi_{1,n}^{(0)}(\xi) = \xi^{\frac12} J_n(n\xi),\quad \wt\phi_{2,n}^{(0)}(\xi)
= \xi^{\frac12} Y_n(n\xi)
\]
with Wronskian
\[
W(\wt\phi_{1,n}^{(0)},\wt\phi_{2,n}^{(0)}) = \frac{2}{\pi}, \quad W(\phi_{1,n}^{(0)},\phi_{2,n}^{(0)}) = \frac{4i}{\pi}
\]
Hence, the forward Green function of~\eqref{eq:master_hom} is
\begin{align*}
G_n(\xi,\xi') &:= \frac{\pi}{4i} [\phi_{1,n}^{(0)}(\xi)\phi_{2,n}^{(0)}(\xi')-\phi_{1,n}^{(0)}(\xi')\phi_{2,n}^{(0)}(\xi)]\chi_{[\xi<\xi']}\\
&= \frac{\pi}{2} [\wt\phi_{1,n}^{(0)}(\xi)\wt\phi_{2,n}^{(0)}(\xi')-\wt\phi_{1,n}^{(0)}(\xi')\wt\phi_{2,n}^{(0)}(\xi)]\chi_{[\xi<\xi']}\\
&= \frac{\pi}{2}(\xi\,\xi')^{\frac12}[ J_n(n\xi)Y_n(n\xi')-
J_n(n\xi')Y_n(n\xi)]\chi_{[\xi<\xi']}
\end{align*}
and thus a basis $\{\phi_{j,n}\}_{j=1}^2$ of \eqref{eq:master} is
given by the Volterra equation
\begin{equation}
  \label{eq:phij_perturb}
  \phi_{j,n}(\xi) = \phi_{j,n}^{(0)}(\xi) + \int_\xi^\infty G_n(\xi,\xi')W_0(\xi',E) \phi_{j,n}(\xi')\,
  d\xi'
\end{equation}
that one now needs to solve. This of course requires a thorough
understanding of the behavior of $J_n(n\xi)$ and
$Y_n(n\xi)$ for large $n$ on intervals of the form $\xi>\xi_0>0$ where $0<\xi_0\ll 1$ is
fixed, see~\cite{AS} and~\cite{Olver}. We leave it to the interested reader to pursue this direction.

\bibliographystyle{amsplain}

\end{document}